\documentclass[11pt]{article}
\usepackage[margin=1in]{geometry}
\usepackage{amsfonts, amsmath, amsthm, amssymb, mathtools}
\usepackage[round]{natbib}

\usepackage{algorithmic}
\usepackage{algorithm}

\usepackage{xcolor}
\usepackage{booktabs}
\usepackage{enumerate}
\usepackage{graphicx,caption}
\usepackage{multirow}
\usepackage{mathtools}
\usepackage{amssymb}
\mathtoolsset{showonlyrefs=true}

\newcommand{\E}{E}
\newcommand{\R}{\mathbb{R}}
\renewcommand{\P}{\mathbb{P}}
\DeclareMathOperator{\pr}{pr}
\DeclareMathOperator{\var}{var}
\DeclareMathOperator{\eff}{eff}
\DeclareMathOperator*{\argmin}{argmin}

\newcommand{\revision}[1]{#1}

\newtheorem{theorem}{Theorem}
\newtheorem{lemma}{Lemma}

\theoremstyle{definition}
\newtheorem{assumption}{Assumption}

\begin{document}

\markboth{S. Khan et~al.}{Doubly-robust and heteroscedasticity-aware sample trimming}

\title{Doubly-robust and heteroscedasticity-aware sample trimming for causal inference}
\author{%
	Samir Khan \\
  Stanford University\\
  \texttt{samirk@stanford.edu} \\
  \and
   Johan Ugander \\
   Stanford University \\
   \texttt{jugander@stanford.edu} \\
}

\maketitle

\begin{abstract}
	A popular method for variance reduction in causal inference is propensity-based trimming, the practice of removing units with extreme propensities from the sample. This practice has theoretical grounding when the data are homoscedastic and the propensity model is parametric \citep{yang2018, crump}, but in modern settings where heteroscedastic data are analyzed with non-parametric models, existing theory fails to support current practice. In this work, we address this challenge by developing new methods and theory for sample trimming. Our contributions are three-fold: first, we describe novel procedures for selecting which units to trim. Our procedures differ from previous works in that we trim not only units with small propensities, but also units with extreme conditional variances. Second, we give new theoretical guarantees for inference after trimming. In particular, we show how to perform inference on the trimmed subpopulation without requiring that our regressions converge at parametric rates. Instead, we make only fourth-root rate assumptions like those in the double machine learning literature. This result applies to conventional propensity-based trimming as well and thus may be of independent interest. Finally, we propose a bootstrap-based method for constructing simultaneously valid confidence intervals for multiple trimmed sub-populations, which are valuable for navigating the trade-off between sample size and variance reduction inherent in trimming. We validate our methods in simulation, on the 2007-2008 National Health and Nutrition Examination Survey, and on a semi-synthetic Medicare dataset and find promising results in all settings. 
\end{abstract}

\section{Introduction}
\label{sec:intro}

Traditional methods for estimating causal effects from observational data typically rely on two standard assumptions: unconfoundedness and overlap \citep{rosenbaum1984}. In practice, observational data often have limited overlap, especially in high-dimensional settings \citep{d2021}, and this leads to extreme propensity scores and high-variance estimates of the treatment effect. A large body of literature addresses this challenge by modifying the estimand to either exclude or down-weight units with extreme propensity scores, and these methods have been widely adopted in practice \citep{yang2018, li2018, crump}. However, modern data can pose an additional challenge in the form of heavy-tails and heteroscedasticity \citep{burke2019, amazon_rct}. 

In this paper, we address this challenge by exploring sample trimming methods that reduce variance by trimming not only units with extreme propensities, but also units with extreme conditional variances. In order to provide valid statistical inferences when using these trimming methods, we also develop new methods for inference after sample trimming that offer greater flexibility and are valid in a wider range of settings than previous such methods. 

\paragraph{Motivation and interpretation} To motivate this approach, consider a single unit $(X, Y, Z)$ drawn from a super-population distribution, where $X$ is a covariate vector, $Y$ is a response, $Z$ is a treatment indicator, and $e(X)$ is the probability of treatment. An inverse-propensity weighted estimator for $\E[Y]$ is $YZ/e(X)$, which is unbiased, but is well-known to suffer from extremely high-variance when $e(X)$ takes small values \citep{hajek, khan2021}. As such, the goal of existing sample trimming methods is to preclude this possibility by removing units for which $e(X)$ takes small values, executing a change of estimand to make what is essentially a bias--variance trade-off. On the other hand, the inverse-propensity weighted estimate $YZ/e(X)$ will also have high-variance if $\var(Y\mid X)$ is large, an issue which is not addressed by existing methods, but may be a major obstacle when $\var(Y\mid X)$ is extremely large for some values of $X$. Put simply: if we do not believe that we can accurately estimate treatment effects on units with propensities of, say, 0.01, then we must also acknowledge that we cannot accurately estimate treatment effects on units with conditional variances of say, 100, and so we propose to trim these latter units as well.


An important difference between this proposal and existing propensity-based trimming methods is that the sub-population found by propensity-based methods can be interpreted as a population that is likely to receive either treatment or control (sometimes called an equipoise population), and thus may be a natural population of interest. This interpretation does not extend to variance-based trimming methods\textemdash instead, variance-based methods can be interpreted as identifying a small population of outliers in the data, whose behavior and response to treatment is very different from that of other units, and trimming these units to focus on an ``inlier'' population on which treatment effects can be estimated more accurately. In many cases, this inlier population is also of natural interest, since treatment effects on the full population may be dominated by the outliers, and the treatment effect on the inlier population may be more representative of how treatment will affect the majority of units. We demonstrate this phenomenon, along with further interpretive issues, as part of a data example in Section~\ref{subsec:acic}. 

\revision{In general, the question of whether or not a particular subpopulation is of interest to an analyst is dependent on both domain considerations and the level of precision with which treatment effects for that subpopulation can be estimated. The problem of selecting a subpopulation of interest from a set of candidates is fundamental to the trimming literature, and not unique to our work\textemdash even when propensity trimming alone, the choice of propensity cut-off induces a similar family of sub-populations, and choosing between those sub-populations requires a similar balancing of variance and relevance. Our methods can be understood as more effectively navigating this trade-off between variance and relevance than existing methods, thus offering practitioners a better set of sub-populations to choose between.}

\paragraph{Inference after trimming}
After applying any sample trimming procedure, another challenge immediately arises: how to perform valid inference on the trimmed sub-population. Thus our second contribution in the present work is to provide new theoretical results on inference after sample trimming. Existing work typically makes strong rate assumptions or parametric assumptions on the estimation of nuisance components \citep{crump, yang2018}, and we extend this work by using doubly-robust estimators to show how valid inference can be performed under weaker conditions on the estimation of nuisance components. The application of doubly-robust estimators to this setting requires a subtle choice of estimand as well as careful handling of cross-fitting, both of which we address. 

These results apply both to our variance-based trimming and to classical propensity-based trimming methods, thus connecting the recent literature on double machine learning and cross-fitting with the long-standing practice of sample trimming.

Our third contribution addresses a more subtle, previously unconsidered, aspect of inference after trimming. Roughly speaking, there are several features of a sample trimming method we may be interested in: the amount of variance reduction offered by the trimming, the size of the resulting sub-population, the point estimate on that sub-population, and perhaps even covariate distributions with the sub-population. However, there is no way to smoothly navigate the trade-offs between these considerations. If we trim the sample one way, perform inference, and find the results unfavorable for some reason, we cannot then trim the sample another way and perform valid inference without conditioning on the results of the first sample trimming; this is the problem of selective inference \citep{taylor2015}. As a remedy, we introduce a bootstrap-based method that allows an analyst to pre-commit to a small number of trimming methods, and then constructs simultaneously-valid confidence intervals for the sub-populations found by each trimming method. An analyst can then choose freely between the different sub-populations based on problem specific considerations while retaining statistical validity. \revision{One drawback of our methods is that, in simulations, we require relatively large sample sizes to obtain the target coverage level, meaning that analysts should be more cautious of results in small sample sizes.}

\paragraph{Response-based trimming}

One potential objection to our approach is that our trimming methods will use the responses when modeling conditional variances, and thus our trimming procedures are response-dependent. This raises two concerns, one statistical and one philosophical. From a statistical perspective, one may be concerned that this compromises the validity of the analysis, but we show in Section~\ref{sec:inference} that under appropriate assumptions on the fitting of the conditional variance, our inferences remain valid despite the fact that we have used the response when trimming. From a philosophical perspective, units with extreme responses may be the units most in need of treatment, and should not be trimmed. However, if these extreme units are actually the ones of most interest, then a measure like the average treatment effect is perhaps not even appropriate, since it will also account for the effect of the treatment on all other units as well. Nonetheless, because our methods provide simultaneously valid confidence intervals across multiple sub-populations, we still provide a point estimate and confidence interval for the average treatment effect on the full population, including any potential units of special importance, when sample trimming. 

To summarize, our work both proposes a new criteria for sample trimming based on conditional variances and propensity scores rather than on propensity scores alone, and develops new theoretical tools for inference after sample trimming. We validate all of our methods with experiments on synthetic, semi-synthetic, and real data and find that our new trimming methods reduce variance beyond what propensity-based methods alone can achieve, identify interesting sub-populations of the full sample by removing possible outliers, and lead to statistically significant conclusions on some of these sub-populations even when no such conclusion was possible on the full population.

\subsection{Related work}

Our work directly builds on the extensive sample trimming literature, and especially on \cite{crump} \revision{(which is itself a journal version of \citet{crump2006})} and \cite{yang2018}. We offer a more detailed comparison with these works in Section~\ref{sec:inference}, but at a high-level, we differ from these previous works in \revision{our more complete treatment of} heteroscedasticity and in assuming weaker conditions on the modeling of nuisance components. \revision{For example, Theorem 1 of \cite{crump} calculates an optimal trimming set in the heteroscedastic case, but then quickly specializes to the homoscedastic setting in Corollary 1, and so the main methodological work is under the homoscedasticity assumption. In contrast, we provide a full methodological toolbox for tackling heteroscedasticity, including allowing for complex nonparametric estimates of conditional variances, and present simultaneous inference methods that can be used to compare subpopulations.} 

One prior work with a similar idea to ours is \cite{chaudhuri2014}, which proposes to remove units whose contribution to the inverse-propensity weighted estimator is extremely large, which also amounts to removing units with extreme response values. However, \cite{chaudhuri2014} are considering a largely different problem than us: they are not concerned with variance minimization, consider only classical inverse-propensity weighted estimators, and do not modify the estimand as is done in the sample trimming literature. 

Our current proposal is also conceptually related to methods in robust statistics and outlier removal. For example, removing units with large residuals from an ordinary least-squares analysis is similar in spirit to the methods we propose here, as are other methods that identify and remove extreme units from the data such as \citet{rohatgi2021}. We differ from these methods in that our motivation for dropping units is based on variance reduction, not on a contamination model for the data, and in that we emphasize the problem of inference after dropping these units.



\section{Model and notation}
\label{sec:model}

We adopt a potential outcomes framework with $n$ units where the tuples $(Y_i(1), Y_i(0), X_i, Z_i)$ are i.i.d.~from a super-population distribution $\P$ over $\mathbb{R}^2\times \mathcal{X}\times \{0,1\}$. We assume that $Y_i(1)$ and $Y_i(0)$ both have finite variance and that we observe $Y_i=Z_iY_i(1)+(1-Z_i)Y_i(0)$. We write $e(x)=\pr(Z_i=1\mid X_i=x)$ for the propensity score, 
$\mu_w(x)=\E[Y_i(w)\mid X_i=x]$, where $w\in\{0,1\}$, for the conditional means, 
and $\sigma^2_w(x)=\var(Y_i(w)\mid X_i=x)$ for the conditional variances. 
We make the standard unconfoundedness and overlap assumptions that $Z_i\perp (Y_i(1), Y_i(0))\mid X_i$ and $\eta\leq e(x)\leq 1-\eta$ \citep{rosenbaum1984}.

Our target of inference is the sample average treatment effect (SATE) and its trimmed analogs, 
\begin{equation}
	\tau=\frac{1}{n}\sum_{i=1}^n \tau(X_i),\quad \tau_{A} =\frac{1}{n_A}\sum_{i=1}^n \tau(X_i)\mathbf{1}\{X_i \in A\},
	\label{eq:target}
\end{equation}
where $\tau(x)=\mu_1(x)-\mu_0(x)$ is a conditional average treatment effect (CATE), $A \subseteq \mathcal X$ is the subset of covariate space we are restricting the covariates to, and $n_A=\sum_{i=1}^n \mathbf{1}\{X_i \in A\}$ is the number of sample units whose covariates lie in $A$. 

In subsequent sections, we employ empirical process notation~\citep{wellner2013, kennedy2016}. We let $W_i = (X_i, Y_i, Z_i)$ be the entire triplet we observe for unit $i$, and we write $\P_nf=\frac{1}{n}\sum_i f(W_i)$ and $\P f= \int f(w)d\P(w)$. Note that for a random function $\hat{f}$, $\P \hat{f}$ is a random variable, since we do not integrate over the randomness in $\hat{f}$. In contrast, $\E[\hat{f}]$ is a deterministic quantity that integrates out the randomness in a new sample and in $\hat{f}$. We also define the norm $\| f\|_{L^q(\P)}=(\P |f|^q)^{1/q}$.


\section{Trimming methods}
\label{sec:methods}

In this section, we present a framework for sample trimming methods and use this framework to propose a trimming method that accounts for conditional variances. As a starting point, recall the result of \cite{hirano2003} that the variance of an efficient estimator \revision{(such as the AIPW estimator)} of $\tau_A$ is given by 
\begin{equation}
	V^{\eff}_{A}=\frac{1}{\pr(X\in A)^2}\E\left[ \mathbf{1}\{X\in A\}\left( \frac{\sigma_1^2(X)}{e(X)}+\frac{\sigma_0^2(X)}{1-e(X)} \right) \right].
\label{eq:eff_var_omega}
\end{equation}
Based on \eqref{eq:eff_var_omega}, we can extract the key quantity that determines a unit's contribution to the asymptotic variance, calling it $k(x)$: 
\begin{equation}
	k(x)=\frac{\sigma_1^2(x)}{e(x)}+\frac{\sigma_0^2(x)}{1-e(x)}.
	\label{eq:k}
\end{equation}
That is, if many units have large values of $k(X_i)$, then the variance of our estimate of $\tau_A$ will be large, and vice-versa. This idea was made precise by \cite{crump}, who showed that, \revision{if $\sigma_0^2(x)$ and $\sigma_1^2(x)$ are bounded}, \eqref{eq:eff_var_omega} is minimized for the set $A$ that thresholds $k(x)$ at a cut-off $\gamma$, that is, 
\begin{equation}
	\argmin_A V_A^{\eff}=\{x: k(x)\leq \gamma\},
	\label{eq:crump_thm}
\end{equation}
for some cut-off $\gamma\in \R$. This result motivates us to consider trimming sets $A$ that have this form, i.e.,~that threshold the function $k(x)$. 

Of course, in practice, we do not have direct access to the function $k$ or the choice of $\gamma$ for which the minimum in \eqref{eq:crump_thm} is attained. Instead, both must be learned from the data, giving us an estimated function $\hat{k}(x)$, an estimated cut-off $\hat{\gamma}$, and a corresponding trimming set $\hat{A}=\{x: \hat{k}(x)\leq \hat{\gamma}\}$. The difference between $\hat{A}$ and $A$ is subtle, but will play a crucial role in what follows, particularly in our discussion of inferential issues in Section~\ref{sec:inference}. We now discuss several choices for $\hat{k}$ and $\hat{\gamma}$.

\subsection{Choices of $\hat{k}$}

How we estimate $k(x)$ depends on what assumptions we are willing to make on $\sigma_1^2(x)$ and $\sigma_0^2(x)$. In particular, we distinguish between two possibilities: 
\begin{description}
	\item[Homoscedasticity assumed:] if we assume that $\sigma_1^2(x), \sigma_0^2(x)$ are constant in $x$ and equal to each other, then we have that $k(x)\propto 1/(e(x)(1-e(x))$, and so we can estimate $k$ by first estimating the propensity score by $\hat{e}(x)$, and then setting $\hat{k}(x)=1/(\hat{e}(x)(1-\hat{e}(x))$, Note that thresholding this choice of $\hat{k}$ is equivalent to thresholding on $\hat{e}(x)$ itself, and so recovers standard propensity trimming \citep{crump}.

			\item[Heteroscedasticity allowed:] if we are not willing to make the homoscedasticity assumption, then we must also estimate the conditional variances by $\hat{\sigma}_1^2(x), \hat{\sigma}_0^2(x)$, and then use the estimate 
				\begin{equation}
					\hat{k}(x)=\hat{\sigma}_1^2(x)/\hat{e}(x)+\hat{\sigma}_0^2(x)/(1-\hat{e}(x)).
					\label{eq:k_hat}
				\end{equation}
\end{description}

Thus, the usual propensity-based trimming corresponds to choosing $\hat{k}$ based on a homoscedasticity assumption that may or may not be satisfied. In some cases, such as when $Y_i$ is binary and so $\sigma^2_w(x)$ is bounded by $1/4$ for all $x$, deviations from this assumption may be negligible. However, in other cases, such as when $Y_i$ is real-valued and has potentially unbounded variance, deviations from this assumption may be significant and worth capturing. In such situations, we propose instead trimming based on the ``heteroscedasticity-aware'' $\hat{k}$ defined in \eqref{eq:k_hat}. This is in contrast to propensity-based trimming, which we refer to as ``homoscedastic trimming'' in light of the underlying homoscedasticity assumption. Going forward, we state all of our results for general $\hat{k}$, making them relevant to both existing (homoscedastic, propensity-based) procedures and our new procedures.

\subsection{Choices of $\hat{\gamma}$}

Next we consider the estimation of $\gamma$. There is a unique choice of $\gamma$ that achieves the minimum in \eqref{eq:crump_thm}, but we may or may not want to target this choice of $\gamma$ depending on how we prioritize other considerations such as simplicity and sub-population size (which can aid in interpretability). Thus we consider three possible choices of $\hat{\gamma}$: 

\begin{description}
	\item[Constant:] The simplest choice is to pre-commit to a specific value of $\hat{\gamma}$; for example, \cite{crump} suggest using $\hat{\gamma}=1/0.1+1/0.9\approx 11.1$ for propensity-based trimming.
	\item[Variance minimizing:] If our goal is to actually attain the minimum in \eqref{eq:crump_thm}, then we should choose $\gamma$  to minimize a sample estimate of $V_{A}^{\text{eff}}$ for $A=\{x: \hat{k}(x)\leq \gamma\}$, 
		\begin{equation}
			\hat{\gamma}=\argmin_{\min_i\hat{k}(X_i)\leq \gamma\leq \max_i\hat{k}(X_i)} \frac{\frac{1}{n}\sum_i \hat{k}(X_i)\mathbf{1}\{\hat{k}(X_i)\leq \gamma\}}{\left( \frac{1}{n}\sum_i \mathbf{1}\{\hat{k}(X_i)\leq \gamma\} \right)^2}.
			\label{eq:gamma_objective}
		\end{equation}
	
	\item[Fixed-fraction:] Finally, if our goal is to trim as little of the sample as possible, so that the remaining study population is as close/relevant as possible to the full population, we may take
		\begin{equation}
			\hat{\gamma}=(1-\delta) \text{ quantile of }\hat{k}(X_1),\ldots, \hat{k}(X_n),
				\label{eq:quantile}
		\end{equation}
		so that we only trim a $\delta$-fraction of the data, for some (presumably modest) constant $\delta$.
\end{description}
Each proposed choice of $\hat{k}$ and $\hat{\gamma}$ represents a particular sample trimming method with different properties. For example, homoscedastic trimming with a constant threshold is very common in practice, but may trim a large fraction of the sample. On the other hand, heteroscedastically trimming a fixed-fraction of the sample lends itself to the outcome-outlier removal interpretation discussed in Section~\ref{sec:intro}, and may substantially reduce variance while trimming only a small number of units. Particular choices will depend on problem-specific considerations, but in this work, we especially highlight the value of heteroscedasticity-aware trimming of a fixed-fraction of units, and focus mainly on this choice in our experiments. 

\section{Inference after trimming}
\label{sec:inference}

After we have chosen a trimming method (i.e., after we have estimated a function $\hat{k}$ and a threshold $\hat{\gamma}$ as in the previous section), we obtain a subset $\hat{A}$ of $\mathcal{X}$ and would like to perform inference around $\tau_{\hat{A}}$. Actually performing such inference turns out to be a fairly subtle task. Although it is not uncommon to simply ignore the trimming step and proceed with standard inference, the resulting confidence intervals are not guaranteed to achieve the desired coverage. On the other hand, existing methods that do account for trimming typically require strong assumptions on the modeling process: for example \cite{crump} require that the outcome regression and propensity both be estimated at least an $o(n^{-1/3})$ rate, while \cite{yang2018} require that the propensity model be parametric and the trimming cut-off $\hat{\gamma}$ be constant. In this paper, we go beyond this prior work by developing a method for valid inference after trimming even when $e(x), \mu_w(x)$, and $\sigma_w^2(x)$ are estimated at the slower $o(n^{-1/4})$ rate. By doing so, we connect the literature on causal inference with machine learning methods to the literature on trimming, allowing practitioners to use doubly-robust estimators on a trimmed sample, and still produce asymptotically valid confidence intervals.

The key to our approach is a careful choice of estimand. In fact, there are two natural estimands around which we might want to perform inference\textemdash the distinction between these has been briefly noted in the literature \citep{crump, yang2018}, but we hope to provide a more detailed discussion. The first possible estimand is the treatment effect on the sample trimming set $\hat{A}$, that is, 
\begin{equation}
	\tau_{\hat{A}}=\frac{1}{n_{\hat{A}}}\sum_{i=1}^n \tau(X_i) \mathbf{1}\{\hat{k}(X_i)\leq\hat{\gamma}\},\quad n_{\hat{A}}=\sum_i \mathbf{1}\{\hat{k}(X_i)\leq\hat{\gamma}\},
	\label{eq:smoothed_estimand}
\end{equation}
where $n_{\hat{A}}$ is the trimmed sample size. However, with this estimand, the specific sub-population on which we are performing inference is dependent on the sample, because $\hat{A}$ is a function of the sample. To obtain a sub-population that is meaningful independent of the realized sample, we suppose that $\hat{k}$ and $\hat{\gamma}$ converge to $\bar{k}$ and $\bar{\gamma}$ respectively (see Assumption~\ref{ass:trim} for a precise statement), and define the limiting sub-population $\bar{A}=\{x: \bar{k}(x)\leq \bar{\gamma}\}$. The estimand corresponding to $\bar{A}$ is 
\begin{equation}
	\tau_{\bar{A}}=\frac{1}{n_{\bar{A}}}\sum_{i=1}^n \tau(X_i) \mathbf{1}\{\bar{k}(X_i)\leq\bar{\gamma}\},\quad n_{\bar{A}}=\sum_{i=1}^n \mathbf{1}\{\bar{k}(X_i)\leq\bar{\gamma}\}.
	\label{eq:limit_smoothed_estimand}
\end{equation}

To better understand the difference between \eqref{eq:smoothed_estimand} and \eqref{eq:limit_smoothed_estimand}, it is helpful to consider the special case of homoscedastic trimming with a fixed value of $\hat{\gamma}$. Then, $\tau_{\hat{A}}$ corresponds to the sub-population trimmed by the estimated propensity $\hat{e}(x)$, while, assuming that $\hat{e}(x)$ is consistent for the true propensity, $\tau_{\bar{A}}$ corresponds to the sub-population trimmed by the true propensity $e(x)$. This special case also highlights the pros and cons of each estimand: the advantage of $\tau_{\hat{A}}$ is that we know exactly which units in our sample are part of the target trimming set $\hat{A}$, while the advantage of $\tau_{\bar{A}}$ is that we can interpret it in terms of the true propensity.

Another important difference between these two estimands is that, when using an augmented inverse propensity weighted estimator \citep{kang2007, robins1994}, performing inference around $\tau_{\hat{A}}$ requires much weaker conditions than performing inference around $\tau_{\bar{A}}$: inference at the usual $n^{1/2}$-rate around $\tau_{\hat{A}}$ requires only a consistency assumption on $\hat{k}$, while inference around $\tau_{\bar{A}}$ requires stronger assumptions on the asymptotics of $\hat{k}$ as well as potentially replacing indicator functions with smooth approximations as in \citet{yang2018} (we discuss the need for this smoothing in more detail following Theorem~\ref{thm:limit_trimmed_clt}). Intuitively, when estimating $\tau_{\hat{A}}$, we do not need to control error terms of the form $\hat{k}(x)-k(x)$ since both the estimator and estimand depend on $\hat{k}$, while we do need to control such terms when estimating $\tau_{\bar{A}}$. This subtle distinction is easy to overlook, but is of central importance to correctly using and interpreting sample trimming methods.

\subsection{Asymptotic linearity of the AIPW estimator}

We now present our main result: that an augmented inverse propensity weighted estimator of $\tau_{\hat{A}}$ is first-order equivalent to a sum of independent and identically distributed random variables for any consistent $\hat{k}$, allowing for inference based on normal theory or the bootstrap. Formally, suppose that we estimate the nuisance parameters $\mu_{w}$ and $e$ by $\hat{\mu}_{w}$ and $\hat{e}$ so that the following assumption is satisfied:
\begin{assumption}
	\label{ass:dml}
	We assume that $\hat{\mu}_{w}$ and $\hat{e}$ are consistent for $\mu_{w}$ and $e$ in the $L^2(\P)$ norm and satisfy the rate condition 
	\begin{equation}
		\| (\hat{\mu}_{w}(X_i)-\mu_{w}(X_i))\mathbf{1}\{\hat{k}(X_i)\leq \hat{\gamma}\}\|_{L^2(\P)}\cdot \|(\hat{e}(X_i)-e(X_i))\mathbf{1}\{\hat{k}(X_i)\leq \hat{\gamma}\}\|_{L^2(\P)}=o_P(n^{-1/2}).
	\label{eq:rate_product}
	\end{equation}
\end{assumption}
In Assumption~\ref{ass:dml}, we only require convergence rates on the subset of non-trimmed units. Thus, even if it is quite difficult to estimate the mean and propensity in some extreme parts of the covariate space, this will not be a problem as long as we trim those parts of the covariate space off. 

Next, we make the following assumption on $\hat{k}$ and $\hat{\gamma}$.
\begin{assumption}
	\label{ass:trim}
	We assume $\hat{k}$ and $\hat{\gamma}$ are convergent, so that $\|\hat{k}-\bar{k}\|_{L^{\infty}(\P)}=o_P(1)$ and $\hat{\gamma}\xrightarrow{\P} {\gamma}$ for some function $\bar{k}:\mathcal{X}\to \R$ and $\bar{\gamma}\in \R$, and that $\hat{k}$ and $\bar{k}$ have densities that are bounded by a constant $B$.
\end{assumption}
Note that we do not assume well-specification of $\hat{k}$, i.e., that $\hat{k}$ as in \eqref{eq:k} actually converges to $\sigma_1^2(x)/e(x)+\sigma_0^2(x)/(1-e(x))$. If this is not true, then we are no longer estimating the optimal trimming, but we are still estimating a well-defined estimand. \revision{However, interpretation of the estimand becomes more challenging in this case, since it is not clear what large values of $\hat{k}$ reflect. Note also that if the conditional variances $\sigma^2_{w}$ are infinite, then Assumption~\ref{ass:trim} will not be satisfied, and so caution should be used in settings where infinite variance is a concern.} 

Finally, we make the following assumption on the estimation of $\hat{\mu}_{w}$, $\hat{e}$, and $\hat{k}$.

\begin{assumption}
	\label{ass:cross}
	We assume that either (a) $\hat{\mu}_{w}$, $\hat{e}$, and $\hat{k}$ are $K$-fold cross-fitted, $K>1$ or that (b) $\hat{\mu}_{w}$ and $\hat{e}$ are restricted to lie in a Donsker class and $\hat{k}$ is restricted to lie in a VC-subgraph class. 
\end{assumption}

\revision{Assumption~\ref{ass:cross} warrants further discussion. In alternative (a), we mean by cross-fitting that the data are split into $K$ equal groups, say $I_1,\cdots, I_K$. Then, for the propensity model for example, for each of $1\leq k \leq K$ a model $\hat{e}^{(-k)}$ is fit on the data from folds $I_1,\cdots, I_{k-1}, I_{k+1},\cdots, I_K$. Finally, the cross-fit model $\hat{e}^{\text{cf}}$ is defined as 
\begin{equation}
	\hat{e}^{\text{cf}}(X_i)=\sum_{k=1}^K \mathbf{1}\{i\in I_k\} \hat{e}^{(-k)}(X_i).
\end{equation}
That is, predictions are made for each point in the sample using the model that did not see that point during training. The cross-fitted version of $\hat{\mu}_{w}$ and $\hat{k}$ are defined similarly. 

In alternative (b), we instead restrict the complexity of the fitted functions $\hat{\mu}_{w}$, $\hat{e}$, and $\hat{k}$. Both of these are standard conditions in the analysis of the doubly-robust estimator \citep{chernozhukov2018, kennedy2016, andrews1994}, although the condition that $\hat{k}$ lie in a VC-subgraph class is stronger than typical conditions; we require this stronger condition to handle the fact that $\hat{k}$ is wrapped in indicator functions for our purposes. 

We present both alternatives in Assumption~\ref{ass:cross} since either is sufficient for the theoretical results of this section, but it is natural to ask which should be preferred in practice\textemdash that is, should we use cross-fitting or not? Based on our experiments, we recommend using a mixture of the two assumptions. In particular, we recommend a procedure in which which $\hat{k}$ is cross-fitted (i.e.~it is constructed using cross-fitted estimates of $e$ and $\sigma_{w}^2$) when selecting which units to trim, but the actual estimation and inference (i.e.~the construction of the AIPW scores in \eqref{eq:aipw_score}) is done using cross-fitted estimates of $e$ and $\mu_{w}$. The intuition behind this recommendation is that cross-fitting $\hat{k}$ complicates interpretation of the estimand, and not cross-fitting $\hat{k}$ ensures that units with extreme values of $\hat{k}$ are correctly identified. Meanwhile, cross-fitting $e$ and $\mu_{w}$ in the estimation step improves the finite-sample coverage properties of our confidence intervals, as it typically does for doubly-robust estimators. Thus, by not cross-fitting $\hat{k}$ but still cross fitting $\hat{e}$ and $\hat{\mu}_w$, we obtain valid confidence intervals around an interpretable estimand.}


With this preparation, we have the following result, whose proof appears in Appendix~\ref{sec:proofs}.
\begin{theorem}
	\label{thm:trimmed_clt}
	Assume that $\hat{\mu}_{w}$, $\hat{e}$, $\hat{k}$, and $\hat{\gamma}$ satisfy Assumptions~\ref{ass:dml}-\ref{ass:cross}. Consider the doubly-robust scores 
	\begin{equation}
		\hat{\tau}(X_i)=\hat{\mu}_1(X_i)+\frac{Y_i(1)-\hat{\mu}_1(X_i)}{\hat{e}(X_i)}Z_i-\hat{\mu}_0(X_i)-\frac{Y_i(0)-\hat{\mu}_0(X_i)}{1-\hat{e}(X_i)}(1-Z_i),
		\label{eq:aipw_score}
	\end{equation}
	and define the doubly-robust estimator 
	\begin{equation}
		\hat{\tau}_{\hat{A}}=\frac{1}{n_{\hat{A}}}\sum_{i=1}^n \hat{\tau}(X_i)\mathbf{1}\{\hat{k}(X_i)\leq\hat{\gamma}\}.
		\label{eq:tau_hat}
	\end{equation}
	Then, we have the asymptotically linear expansion
	\begin{equation}
		\hat{\tau}_{\hat{A}}-\tau_{\hat{A}}=\frac{1}{\pr(\bar{k}(X_i)\leq\bar{\gamma})}\cdot \frac{1}{n}\sum_{i=1}^n \left(\frac{Y_i(1)-\mu_1(X_i)}{e(X_i)}Z_i+\frac{Y_i(0)-\mu_0(X_i)}{1-e(X_i)}(1-Z_i)\right)\mathbf{1}\{\bar{k}(X_i)\leq\bar{\gamma}\}+o_P(n^{-1/2}).
		\label{eq:trimmed_clt}
	\end{equation}
\end{theorem}
The value of Theorem~\ref{thm:trimmed_clt} is that the right-hand side of \eqref{eq:trimmed_clt} is an i.i.d.~sum. This means, \revision{assuming that $Y_i(w)-\mu_w(X_i)$ has finite variance}, that we have the central limit theorem, 
\begin{equation}
	\sqrt{n}(\hat{\tau}_{\hat{A}}-\tau_{\hat{A}})\xrightarrow{d}N\left( 0, V \right),\quad V = \frac{1}{\pr(\bar{k}(X_i)\leq \bar{\gamma})^2}\E\left[ \mathbf{1}\{\bar{k}(X_i)\leq \bar{\gamma}\}k(X_i) \right],
	\label{eq:var_expression}
\end{equation}
and so if we estimate $V$ using the sample variance of the terms of the sum on the RHS of \eqref{eq:tau_hat}, we can form a confidence interval for $\tau_{\hat{A}}$ using asymptotic normal theory. Similarly, we can bootstrap the terms of the sum on the RHS of \eqref{eq:tau_hat} to obtain bootstrap confidence intervals. 

Our next result is a negative one, and shows that the situation for $\tau_{\bar{A}}$ is less simple. To state it, we let $$\hat{S}_{\hat{A}}=\sum_{i=1}^n \hat{\tau}(X_i)\mathbf{1}\{\hat{k}(X_i)\leq \hat{\gamma}\},\quad S_{\bar{A}}=\sum_{i=1}^n \tau(X_i)\mathbf{1}\{\bar{k}(X_i)\leq \bar{\gamma}\},$$ be the numerators of $\hat{\tau}_{\hat{A}}$ and $\tau_{\bar{A}}$ respectively, so that $\hat{\tau}_{\hat{A}}=\hat{S}_{\hat{A}}/n_{\hat{A}}$ and $\tau_{\bar{A}}=S_{\bar{A}}/n_{\bar{A}}$. The usual approach to establishing a central limit theorem for $\hat{\tau}_{\hat{A}}$ would be to first establish a joint central limit theorem for $\hat{S}_{\hat{A}}$ and $n_{\hat{A}}$, and then use the delta method obtain a central limit theorem for $\hat{\tau}_{\hat{A}}-\tau_{\bar{A}}$. To illustrate the challenges of this approach, the following theorem assumes that such a procedure could be carried out, and shows what the result would be. 

\begin{theorem}
	\label{thm:limit_trimmed_clt}
	Suppose Assumptions~\ref{ass:dml}-\ref{ass:cross} hold, and let $\hat{\tau}(X_i)$ be as in \eqref{eq:aipw_score}. Then, if $n^{-1/2}([\hat{S}_{\hat{A}}\;\; n_{\hat{A}}]-[S_{\bar{A}}\;\; n_{\bar{A}}])$ converges in distribution, we have that 
	\begin{equation}
		\hat{\tau}_{\hat{A}}-\tau_{\bar{A}}=\frac{1}{\pr(\bar{k}(X_i)\leq \bar{\gamma})}\left(\frac{1}{n}\sum_{i=1}^n \left(\frac{Y_i-\mu_1(X_i)}{e(X_i)}Z_i+\frac{Y_i(0)-\mu_0(X_i)}{1-e(X_i)}(1-Z_i)\right)\mathbf{1}\{\bar{k}(X_i)\leq\bar{\gamma}\}+\Delta_n\right),
\label{eq:limit_trimmed_clt}
	\end{equation}
	for 
	\begin{equation}
		\Delta_n=\P\left[ \left(\tau(X_i)-\E[\tau(X_i)\mid \bar{k}(X_i)\leq \bar{\gamma}]\right)\left(\mathbf{1}\{\hat{k}(X_i)\leq \hat{\gamma}\}-\mathbf{1}\{\bar{k}(X_i)\leq \bar{\gamma}\}\right) \right]+o_P(n^{-1/2}).
	\end{equation}
\end{theorem}
The convergence in distribution assumption above essentially amounts to ensuring that the error term in the delta method/Taylor expansion can be controlled; the point of the theorem is that even if this is true, we do not obtain an asymptotically linear expansion without further assumptions. 

To elaborate, the expansion of $\hat{\tau}_{\hat{A}}-\tau_{\bar{A}}$ in Theorem~\ref{thm:limit_trimmed_clt} contains not only an i.i.d.~sum, but also an error term $\Delta_n$. This $\Delta_n$ depends on $\mathbf{1}\{\hat{k}(X_i)\leq \hat{\gamma}\}-\mathbf{1}\{\bar{k}(X_i)\leq \bar{\gamma}\}$, and so the asymptotics of $\Delta_n$ depend on the asymptotics of $\hat{k}$ itself, preventing the use of black-box machine learning algorithms to fit $\hat{k}$. Furthermore, even in a case where the asymptotics of $\hat{k}$ are well-understood (for example, if $k$ is assumed to follow a parametric model and $\hat{k}$ is fit by maximum-likelihood estimation), analyzing $\Delta_n$ is still challenging because $\hat{k}(X_i)$ is wrapped in non-smooth indicators, preventing the use of standard asymptotic tools like the delta method. 

One way to overcome these challenges is developed by \citet{yang2018}, who assume that $\hat{k}$ is fit using a generalized linear model and then replace $\tau_{\hat{A}}$ and $\tau_{\bar{A}}$ with smoothed equivalents. With this smoothing, the usual asymptotic theory is once again applicable, and they then track the contribution of $\Delta_n$ to obtain a central limit theorem for $\hat{\tau}_{\hat{A}}-\tau_{\bar{A}}$ in the case of homoscedsatic trimming with a fixed cut-off. However, we do not develop such procedures here for two reasons. First, they are applicable only for specific choices of $\hat{k}$ and thus do not allow for flexible non-parametric modeling of the propensities and outcome model. Second, in replacing $\tau_{\bar{A}}$ with a smoothed analog, we are effectively performing a second change of estimand on top of the original change of estimand from $\tau$ to $\tau_{\bar{A}}$. This double estimand change complicates the interpretation of any results, and we prefer to avoid it.

In principle, it is possible to construct other estimators for $\tau_{\bar{A}}$ that can be analyzed without making parametric assumptions on $\hat{k}$, e.g., by finding an efficient influence function and constructing a one-step estimator \citep{kennedy2022}, but this approach may require estimating further moments beyond the conditional mean and variance, complicating the problem significantly.

Taken together, Theorem~\ref{thm:trimmed_clt} and Theorem~\ref{thm:limit_trimmed_clt} lead us to recommend interpreting the results of a sample-trimmed analysis in terms of $\hat{\tau}_{\hat{A}}$, since this approach allows for flexible non-parametric modeling, and is thus the approach we take throughout our experiments. 

\section{Simultaneous trimming}
\label{sec:simul}

The results of Section~\ref{sec:inference} allow us to perform inference around a single trimmed sub-population; in this section, we develop methods for studying treatment effects around multiple trimmed sub-populations with simultaneous coverage guarantees. \revision{Such methods often lead to large losses in power, but we argue below and demonstrate empirically in Section~\ref{sec:experiments} that the sample trimming problem has special structure that ameliorates this issue. Thus, the use of simultaneous inference methods is especially appealing in this setting, since it is both useful and, from a power perspective, cheap.}

To illustrate the problem, consider an analyst who decides to use one of the sample trimming methods from Section~\ref{sec:methods}; for example, heteroscedastic trimming with a fixed cut-off $\hat{\gamma}$ as in \eqref{eq:quantile}. Suppose that, after performing inference with the chosen cut-off $\hat{\gamma}$, the analyst finds that the width of the confidence interval for the treatment effect on the sub-population is essentially the same as the width of the confidence interval for the treatment effect on the full population, i.e.,~the trimming had almost no effect. In this case, the analyst would likely want to try another, larger, cut-off $\hat{\gamma}$, and see if the corresponding sub-population has more favorable variance properties.

However, it is now extremely challenging to produce valid confidence intervals for other cut-offs $\hat{\gamma}$ because of the selection bias introduced by the first step of the analysis. Essentially, any analysis conducted after finding the first trimming to be undesirable would have to be carried out conditional on this fact. There is a growing body of literature that develops selective inference methods to address problems of this kind \citep{taylor2015}, but it is somewhat challenging to develop such methods in this context while still allowing for flexible modeling of $\hat{k}$. 

As an alternative, we propose a method that allows an analyst to pre-commit to several different trimming methods in advance, and then produces a confidence intervals for each method such that the confidence intervals are simultaneously valid. That is, if we construct 95\% confidence intervals, we ensure that the probability that \emph{any} interval fails to cover is at most 5\%. We focus our discussion on heteroscedastic trimming with a fixed-fraction cut-off $\hat{\gamma}$, since this is a case where simultaneous intervals need not be much wider than marginal intervals, for reasons discussed below, and also since this trimming method contains a natural hyperparameter $\delta$, the fraction of the sample trimmed, that an analyst may wish to explore. The same procedure could also be used in other ways, for example, to explore homoscedastic trimming with a range of several propensity cut-offs.

The method is based on \citet{beran1988} and an exposition of it also appears as Example 15.4.4 of \citet{lehmann2005}, \revision{with a proof of validity given in Theorem 15.4.6}. The idea is as follows: if we specify a grid of trimming fractions $\{\delta_1,\cdots, \delta_m\}$ in advance, learn the thresholding function $\hat{k}$, and then choose cut-offs $\hat{\gamma}_1,\cdots, \hat{\gamma}_m$ according to \eqref{eq:quantile}, this procedure gives rise to a set of sub-populations $\hat{A}_1,\cdots, \hat{A}_m$. Then, the problem of building confidence intervals that are simultaneously valid for all of these sub-populations is equivalent to building an $\ell^{\infty}$ confidence region for the vector 
$$\left[\begin{array}{cccc} \hat{\tau}_{\hat{A}_1}& \hat{\tau}_{\hat{A}_2}&\cdots& \hat{\tau}_{\hat{A}_m}\end{array}\right].$$
But it follows from Theorem~\ref{thm:trimmed_clt} that this vector is asymptotically normal, and so we can construct such a confidence region using the bootstrap \citep{lehmann2005}. \revision{We bootstrap the Student's $t$-statistic for each parameter, so that the width of our confidence interval is not dominated by the variance of the highest-variance estimate. For more details on the practice of bootstrapping $t$-statistics, we refer the interested reader to Section 12.5 of \citet{efron1994}.} The key step is that, rather than taking the quantile of the individual bootstrap replicates, we take the quantiles of the maximum of the bootstrap replicates across all sub-populations. This procedure is described in detail in Algorithm~\ref{alg:simul_trim}. Note that, in each bootstrap iteration, we do not re-fit the propensity, mean, and conditional variance estimates; this is justified by the arguments in Appendix C.10.1 of \cite{dorn2021}.

\begin{algorithm}[t]
	\caption{Simultaneous trimming}\label{alg:simul_trim}
	\textbf{Input:} data $(X_1, Z_1, Y_1),\cdots, (X_n, Z_n, Y_n)$, trimming fractions $\delta_1,\cdots, \delta_m$, and confidence level $\alpha$\\
\textbf{Output:} simultaneously valid confidence intervals for trimmed sub-populations for $1\leq j\leq m$
\begin{algorithmic}[2]
	\STATE Estimate $e(x)=\pr(Z=1\mid X=x)$ by $\hat{e}$, $\sigma^2_w(x)=\var(Y(w)\mid X=x)$ for $w=0,1$ by $\hat{\sigma}^2_w$ and set $\hat{k}(x)$ as in \eqref{eq:k_hat}
\STATE Set $\hat{\gamma}_j$ to be the $1-\delta_j$ quantile of $\hat{k}(X_1),\cdots, \hat{k}(X_n)$ and let $\hat{A}_j=\{x: \hat{k}(x)\leq \hat{\gamma}_m\}$
\STATE Compute  $\hat{\tau}_{\hat{A}_j}$ as in in Theorem~\ref{thm:trimmed_clt} as well as standard errors $\hat{V}_j$ as in \eqref{eq:var_expression}
\FOR{$b=1,\cdots, B$}
\STATE Draw a bootstrapped dataset $(\tilde{X}_1, \tilde{Z}_1, \tilde{Y}_1),\cdots, (\tilde{X}_n, \tilde{Z}_n, \tilde{Y}_n)$
\STATE Compute point estimates $\hat{\tau}_{\hat{A}_j^b}^b$, and standard errors $\hat{V}_j^b$ on the bootstrapped sample (without refitting $\hat{k}$)
\STATE Compute $$T_b=\max_{1\leq j\leq m}\left|\frac{\hat{\tau}_{\hat{A}_j^b}-\hat{\tau}_{\hat{A}_j}}{\sqrt{\hat{V}^b_j}}\right|$$
\ENDFOR
\STATE Let $q$ be the $1-\alpha$ quantile of $T_1,\cdots, T_B$\\
\RETURN $\hat{\tau}_{\hat{A}_j}\pm q\sqrt{\hat{V}}_j$ for $1\leq j\leq m$
\end{algorithmic}
\end{algorithm}

Of course, the intervals produced by Algorithm~\ref{alg:simul_trim} will be wider than the intervals produced for a single estimand, since this is the price of the simultaneous validity. However, there are two things we can do to ensure they are \revision{not} much wider: choosing a small value of $m$, so that we are not attempting to cover many estimands simultaneously, and choosing values of $\delta_j$ that are close together so that the estimates are strongly positively correlated. 

\revision{The value of the positive correlation is that, in Algorithm 1, we are estimating a quantile of the distribution of the maximum of several approximately standard normal random variables. If these random variables were independent or negatively correlated, the distribution of the maximum would be stochastically larger (since it would be more likely that there is at least one random variable which takes a large value). The positive correlation ensures that if one random variable takes a small value, the others are likely to take small values as well, and so the distribution of the maximum is stochastically smaller. It is also constructive to consider the degenerate case, in which the random variables are all equal to each other and thus perfectly correlated, in which case the desired quantile is the usual 1.96 from a single standard normal. As we deviate further from this degenerate perfect correlation, the estimated quantile correspondingly inflates.}

In our experiments, we find that the grid $\{0, 0.1, 0.2, 0.3\}$ gives reasonable results. This approach is informed by the interpretation of heteroscedsatic trimming as outlier detection, and allows an analyst to determine whether a large amount of variance in the estimator is being introduced by a handful of units. In these cases, our method produces much lower variance estimates on large sub-populations, which may be of interest.

\section{Experiments}
\label{sec:experiments}

We now present a series of experiments to evaluate our proposed trimming methods. 

Our first set of experiments is a coverage study on simulated data to confirm that intervals produced using the methods of Sections~\ref{sec:inference} and \ref{sec:simul} obtain the target coverage level in large samples. Our second set of experiments is conducted on the National Health and Nutrition Examination Survey (NHANES) \citep{nhanes}, which has been studied in previous sample trimming work as well \citep{yang2018, hsu2013}, and provides a comparison between our methods and traditional (propensity-based) trimming methods. 

Our third set of experiments involves data from the 2022 American Causal Inference Conference (ACIC) data challenge \citep{acic_challenge}, which is a semi-synthetic dataset designed to mimic a real dataset of Medicare interventions. We use this data to demonstrate the value of the simultaneous trimming method of Section~\ref{sec:simul} in particular, constructing confidence intervals around several sub-populations simultaneously. We see that the confidence intervals for some of the sub-populations are much narrower than those around the full population. Most importantly, having set aside concerns of selective inference, we can freely choose which of these populations to report the effect on based on our preferences between confidence interval width and sample size.

\subsection{Coverage experiments}
\label{subsec:coverage}

\revision{We begin by studying the coverage rates of confidence intervals for the estimators in Section~\ref{sec:inference}. To do so, we generate data from the following model: 
\begin{equation}
	X_1, X_2, X_3\sim N(0,1),\quad Y(0)=2X_1-X_2+N(0, 1+\max(0, X_2))+\epsilon,\quad Y(1)=Y(0)+\tau(X),
	\label{eq:dgp}
\end{equation}
where the propensity score and treatment effect maps are
\begin{equation}
	\tau(X)=2\mathbf{1}\{X_3<0\}+\mathbf{1}\{X_3>0\}+5X_2\quad\text{and}\quad e(X)=\frac{1}{1+2\exp(X_2-X_1)}
\end{equation}
and the noise $\epsilon$ has a t-distribution with 5 degrees of freedom.

This is a straightforward linear model with two important features: first, there is heteroscedasticity in the distribution of $Y(0)$, meaning that there is in fact a non-trivial conditional variance to model, and the propensity $e(X)$ does not come from a standard logistic model, necessitating the use of non-parametric methods. We threshold the propensities to lie in $[0.05, 0.95]$\textemdash this is to obtain convergence to asymptotic coverage rates in reasonable sample sizes. With more extreme propensities, we expect to find similar results, but in possibly larger sample sizes than those considered here. Then, we estimate the propensities and conditional means using the $\texttt{regression\_forest}$ function from the R package $\texttt{grf}$ \citep{grf}; we fit conditional variances by subtracting estimates of $\E[Y_i\mid X_i]^2$ and $\E[Y_i^2\mid X_i]$ \revision{(this procedure may theoretically produce negative values, but does not in our case; if it does, we recommend instead regressing $X_i$ onto $Y_i$, and then regressing the squared residuals onto $X_i$ with a range-bounded regression method, which will always produce non-negative estimates). Following the discussion after Assumption~\ref{ass:trim}, we do not cross-fit $\hat{k}$ when selecting which units to trim, but do cross-fit when constructing the AIPW scores.}

Based on these estimates, we heteroscedastically trim a $\delta$ fraction of the data for $\delta = 0, 0.05, 0.1$. For each of these sub-populations, we construct marginal 95\% confidence intervals using the central limit theorem in \eqref{eq:trimmed_clt} implied by Theorem~\ref{thm:trimmed_clt}, as well as simultaneous 95\% confidence intervals for all three sub-populations using Algorithm~\ref{alg:simul_trim}.  We can estimate the coverage rates of these confidence intervals for the estimands $\tau_{\bar{A}}$ and $\tau_{\hat{A}}$ across multiple trials; results are shown in Table~\ref{tab:coverage}.

\begin{table}
	\centering
	\begin{tabular}{lllllll}
	\toprule
	&&\multicolumn{5}{c}{Sample size, $n$}\\
	\cmidrule{3-7}
	Estimand&Trimming&500&1000&2000&4000&8000\\
	\midrule

\multirow{4}{*}{$\tau_{\hat{A}}$}&$\delta=0.1$&0.962&0.949&0.949&0.97&0.967\\
&$\delta=0.05$&0.956&0.938&0.944&0.962&0.959\\
&$\delta=0$&0.934&0.929&0.933&0.962&0.96\\
&Simultaneous&0.958&0.937&0.94&0.963&0.959\\[0.15in]

\multirow{4}{*}{$\tau_{\bar{A}}$}&$\delta=0.1$&0.927&0.93&0.913&0.903&0.871\\
&$\delta=0.05$&0.936&0.921&0.92&0.931&0.919\\
&$\delta=0$&0.937&0.93&0.934&0.962&0.958\\
&Simultaneous &0.927&0.91&0.916&0.916&0.894\\
\bottomrule
	\end{tabular}
	\caption{Coverage of 95\% confidence intervals based on Theorem~\ref{thm:trimmed_clt} for treatment effects when heteroscedastically trimming \revision{0\%}, 5\%, and 10\% of the sample on data generated from \eqref{eq:dgp} for different sample sizes, as well as coverage of simultaneous confidence intervals constructed using Algorithm~\ref{alg:simul_trim}. All results are based on 1000 trials.}
	\label{tab:coverage}
\end{table}

\begin{table}
	\centering
	\begin{tabular}{lllllll}
	\toprule
	&&\multicolumn{5}{c}{Sample size, $n$}\\
	\cmidrule{3-7}
	&Trimming&500&1000&2000&4000&8000\\
	\midrule
	&$\delta=0.1$&1.976&1.477&1.105&0.819&0.608\\
&$\delta=0.05$&1.937&1.452&1.091&0.812&0.606\\
&$\delta=0$&1.903&1.436&1.085&0.813&0.61\\
&0.1/0.9 rule&1.914&1.468&1.127&0.85&0.636\\
	\bottomrule
	\end{tabular}
	\caption{Width of 95\% confidence intervals based on Theorem~\ref{thm:trimmed_clt} for treatment effects when heteroscedastically trimming \revision{0\%}, 5\%, and 10\% of the sample on data generated from \eqref{eq:dgp} for different sample sizes, as well as when trimming based on the 0.1/0.9 rule of \citet{crump}. All results are based on 1000 trials.}
	\label{tab:width}
\end{table}

These results confirm the predictions made by our theory: the coverage of our intervals for $\tau_{\hat{A}}$ reaches the target 95\% level in moderately large sample sizes, while the coverage of our intervals for $\tau_{\bar{A}}$ does not, since inference around this estimand is more challenging, as discussed in Section~\ref{sec:inference}. Of course, for $\delta=0$, there is no trimming, and so the two estimands coincide, and the desired coverage is attained for both. Furthermore, our simultaneous intervals also attain the desired coverage rates, showing that the bootstrap procedure of Section~\ref{sec:simul} has the desired properties. 

Finally, we also show in Table~\ref{tab:width} the widths of our intervals compared to intervals obtained by the homoscedastic trimming of Crump et al.~,which we refer to from here on as the 0.1/0.9 rule, since it discards units with propensities outside of this range. The differences in this simulation are small, since the amount of heteroscedasticity is not large, but we see that in larger sample sizes, intervals based on heteroscedastic trimming with $\delta=0.05$ are narrower than both the Crump intervals and the intervals obtained without any trimming. }


\subsection{NHANES experiments}
\label{subsec:nhanes}

Following prior work on sample trimming methods, we analyze the 2007--2008 U.S. National Health and Nutrition Examination Survey data \citep{nhanes}, attempt to estimate the effect of smoking on blood lead levels as in \citet{yang2018, hsu2013}, and compare our trimming methods to other common trimming methods. The dataset contains 3,340 subjects, 679 of whom are smokers, and 2,661 of whom are not. The response, blood lead level, ranges from 0.18 $\mu{}$g/dl to 33.10 $\mu{}$g/dl, in the population, suggesting that there may be heteroscedasticity in the data. As before, we fit the propensities and conditional means using the $\texttt{regression\_forest}$ function from the R package $\texttt{grf}$ \citep{grf}; we fit conditional variances by subtracting estimates of $\E[Y_i\mid X_i]^2$ and $\E[Y_i^2\mid X_i]$ (this may again theoretically produce negative values, but does not in our case, suggesting that the regressions are fairly accurate). \revision{All regressions use gender, education level, income, age, and race as covariates, as in \citet{yang2018}. We do not cross-fit our regression when selecting which units to trim, but do cross-fit when constructing point estimates and confidence intervals, as discussed in Section~\ref{sec:inference}.}



We compare three different sample trimming methods on this data to the baseline of no sample trimming. In the language of Section~\ref{sec:methods}, these are homoscedastic trimming with the constant cut-off of 11.1 (recall that $1/0.1+1/0.9\approx 11.1$, so this is the 0.1/0.9 rule proposed by \cite{crump}), homoscedastic trimming with a variance-minimizing cut-off, and heteroscedastic trimming with variance-minimizing cut-off. For all of these methods, we perform inference using the estimator of Theorem~\ref{thm:trimmed_clt}. Point estimates and 95\% confidence intervals constructed from asymptotic normal theory are shown in Figure~\ref{fig:nhanes}. We note that, for the optimal homscedastic trimming, the implied propensity cut-offs found are $0.04$ and $0.96$. 

\begin{figure}[t]
	\centering
	\includegraphics[width=\textwidth]{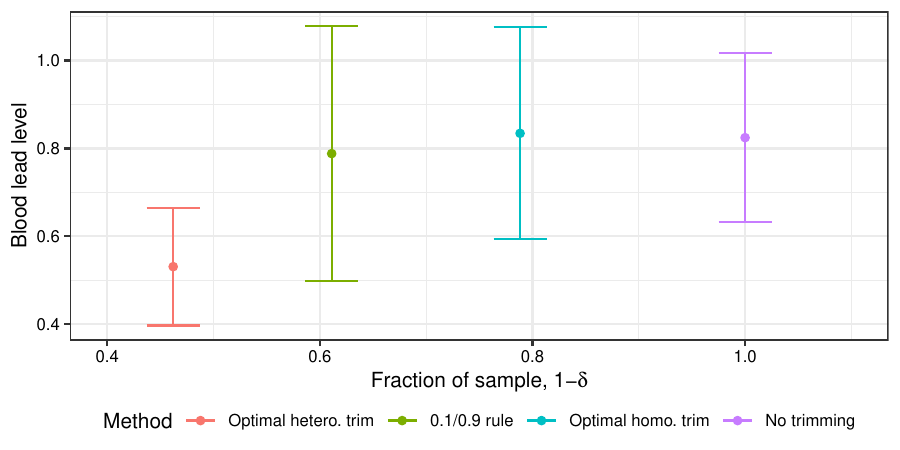}
	\caption{Normal theory 95\% confidence intervals for estimated treatment effects using three different sample trimming methods: optimal heteroscedasticity-aware trimming (red), optimal homoscedastic trimming (blue), the 0.1/0.9 rule of \cite{crump} (green) and also no trimming (purple).} 
	\label{fig:nhanes}
\end{figure}

\begin{figure}[t]
	\centering
	\includegraphics[width=\textwidth]{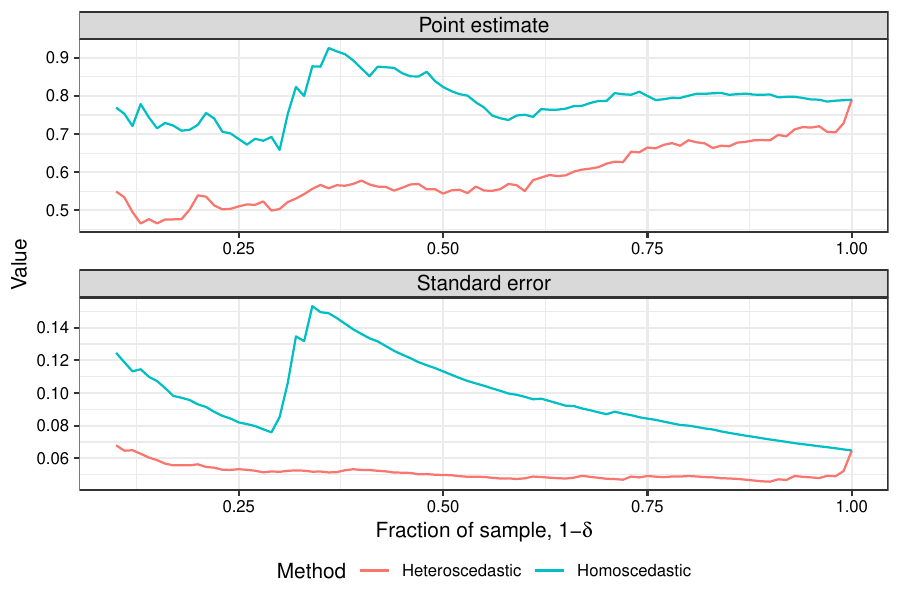}
	\caption{Point estimates (top panel) and standard errors (bottom panel) when trimming a fixed-fraction $\delta$ of the sample using either the homoscedastic (blue) or heteroscedastic (red) approach for a range of values of $\delta$.}
	\label{fig:nhanes_seq}
\end{figure}

We see from these results that, in this problem, sample trimming based only on propensities offers little variance reduction\textemdash the variance-minimizing homoscedastic trimming and the 0.1/0.9 rule both actually slightly increase the variance of our point estimate. On the other hand, the variance-minimizing heteroscedastic trimming offers a 17\% reduction in variance and visibly smaller confidence intervals, as seen in Figure~\ref{fig:nhanes}. 

To further underscore the importance of modeling heteroscedsaticity in this problem, we also plot the full path of point estimates and standard errors as a function of the fraction of the sample trimmed in Figure~\ref{fig:nhanes_seq}. In the language of Section~\ref{sec:methods}, this figure shows the result of trimming a fixed-fraction $\delta$ of the sample homoscedastically or heteroscedastically as $\delta$ ranges from 0 to 1. (Of course, the plot of standard errors in Figure~\ref{fig:nhanes_seq} is not the function that the optimal $\hat{\gamma}$ is chosen to minimize, since these are standard errors computed using the true responses; see Appendix~\ref{app:acic} for a further discussion of this issue.) 

We see that, when modeling heteroscedasticity, the standard error of our estimates decreases as we trim problematic units, and then begins to increase once our sample size becomes very small, which is essentially the expected behavior. On the other hand, when using the homoscedasticity assumption, the standard error of our estimate is more erratic and shows no consistent behavior, because we have failed to fully capture the variance structure of the problem. 

We can also see from Figure~\ref{fig:nhanes} that the estimated treatment effect for the sub-population found by the heteroscedastic trimming is quite different from the estimated treatment effects for the sub-populations found by other trimming methods. This is true even though the sub-population found by optimal heteroscedastic trimming is only about 15\% smaller than the sub-population found by the 0.1/0.9 rule. It is also to be expected, since units with high conditional variances are intuitively likely to have high conditional means as well, and thus dropping them from the sample may substantially change the treatment effect. This change in treatment effect also means that it is important to understand the sub-population we are trimming down to and how it differs from the full population. We explore these differences by analyzing the distribution of the education covariate in the full sample, the variance-minimizing heteroscedsatic sub-population, and the variance-minimizing homoscedastic sub-population, shown in Table~\ref{tab:nhanes_covariates}.

\revision{Table~\ref{tab:nhanes_covariates} shows that the distribution of education level in the population selected by the variance-minimizing heteroscedastic trimming is similar to the distribution of education level in the population selected by the variance-minimizing homoscedastic trimming, except that the heteroscedastic trimming more aggressively trims units with no high school education and less aggressively trims units with high school education. This behavior can be understood by looking at the average values of our fitted propensities and variances within each education level, which are also shown in Table~\ref{tab:nhanes_covariates}. These reveal that the estimated conditional variance for units with no high school is very high, which is why the heteroscedastic trimming trims them more aggressively, and prefers units with high school education instead. The net effect is that the population identified by the heteroscedastic trimming is, on average, more educated than the population identified by the homoscedastic trimming.}


Thus, in this application, trimming the sample based on propensities alone is not able to provide any meaningful variance reduction, but heteroscedasticity-aware sample trimming is able to identify a sub-population on which the variance of the estimated treatment effect is smaller than the variance of the estimated treatment effect on the full population. Further analysis of this sub-population shows that it is roughly a sub-population of more educated individuals. Of course, whether or not this sub-population is actually of interest will depend on analyst-specific considerations, and our goal here is not to suggest that this sub-population is necessarily preferable to the full population. Indeed, the question of whether inference around a sub-population is relevant to a given domain problem is central to all trimming methods. Even the propensity-trimmed sub-population, with its down-weighting of more educated individuals, represents a sub-population that may or may not be preferable to the full population. Instead, we hope this experiment sheds light on how our proposed methods compares to existing methods, and clarifies what they are able to offer.

\begin{table}
	\centering
	\begin{tabular}{llllllll}
		\toprule
		&\multicolumn{3}{c}{\% of sample}&&\multicolumn{3}{c}{Avg.}\\
	\cmidrule{2-4} \cmidrule{6-8}
		Education level &Original&$\hat{A}^{\text{het}}$&$\hat{A}^{\text{hom}}$&&$\hat{e}$&$\hat{\sigma}_0$&$\hat{\sigma}_1$\\
		\midrule
		No HS&13.3&5.8&10.8&&0.166&1.912&6.349\\
		Some HS&16.5&19.9&20.9&&0.29&1.992&4.352\\
		HS&25.1&33.8&28.6&&0.256&1.306&3.305\\
		Some college&25.4&28.8&28.4&&0.212&1.321&3.48\\
		College&19.6&11.7&11.3&&0.079&1.302&2.828\\
\bottomrule
	\end{tabular}
	\caption{Distribution of education level in the full sample of the NHANES data, in the optimal variance-aware trimming set $\hat{A}^{\text{het}}$, and in the optimal propensity trimming set $\hat{A}^{\text{hom}}$, along with average values of fitted models for $\hat{e}$, $\hat{\sigma}_0$, and $\hat{\sigma}_1$ within each education level. We see, for instance, that $\hat{A}^{\text{het}}$ includes fewer people with some high school than $\hat{A}^{\text{hom}}$ does, because even though their estimated propensities are not very small, their estimated variances are very large.}
	\label{tab:nhanes_covariates}
\end{table}

\subsection{ACIC experiments}
\label{subsec:acic}

\begin{figure}[t]
	\centering
	\begin{tabular}{cc}
		\includegraphics[width=0.5\textwidth]{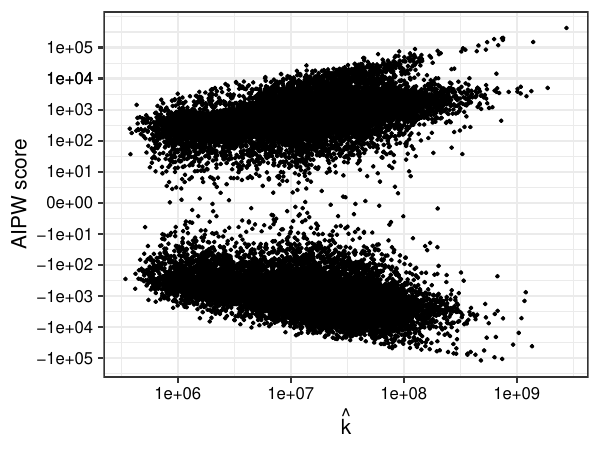}&\includegraphics[width=0.5\textwidth]{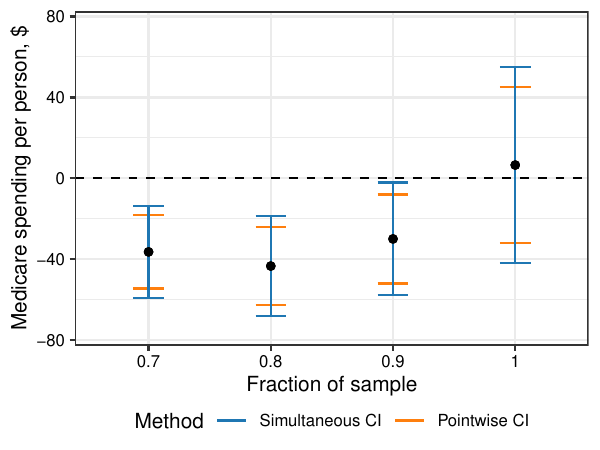}
	\end{tabular}
	\caption{Left: estimated trimming function $\hat{k}$ versus AIPW scores (the terms of the sum in \eqref{eq:tau_hat}); right: pointwise (orange) and simultaneous (blue) confidence intervals on subsets of different sizes of the ACIC data.}
	\label{fig:acic}
\end{figure}

\begin{table}
	\centering
	\begin{tabular}{lllllll}
		\toprule
		Sample frac.&$\hat{\tau}$&$\hat{\sigma}$&Pointwise CI&Simultaneous CI\\
		\midrule

		1&6.551&19.813&$(-31.941, 45.043)$&$(-41.966, 55.068)$\\
		0.9&-29.952&11.4&$(-51.951,-7.953)$&$(-57.867,-2.038)$\\
		0.8&-43.379&10.129&$(-62.591,-24.166)$&$(-68.182,-18.575)$\\
		0.7&-36.382&9.255&$(-54.459,-18.305)$&$(-59.044,-13.72)$\\
		\bottomrule

	\end{tabular}
	\caption{Pointwise and simultaneous confidence intervals on subsets of different sizes of the ACIC data. The treatment has a small positive effect on the full population, but this is only due to a few outliers. In fact, the effect of the treatment is negative on large sub-populations of the data, and the simultaneous confidence intervals show that this negative effect is significant on the sub-populations of 80\% and 70\% of the data, even though the effect of the treatment is not significant on the full population. Because our confidence intervals and uniformly valid, we may choose which population to report the effect on, or simply report all four.}
	\label{tab:acic}
\end{table}
Our final set of experiments involves data from the 2022 American Causal Inference Conference's data challenge \citep{acic_challenge}. This dataset is a semi-synthetic dataset constructed to mimic data from evaluations of Medicare interventions on the U.S. health-care system. \revision{The treatment in this data set is a binary intervention on the process by which Medicare covers the cost of medical expenses for patients, although the exact intervention has not been disclosed to avoid releasing sensitive information. The outcome is the amount of money Medicare spends on a given patient, and is extremely heavy-tailed, with some patients requiring a few hundred dollars of Medicare spending and others requiring tens of thousands of dollars of spending. In such a setting, trimming the sample to exclude some outlier patients is likely to offer significant variance reduction.}

The data consists of observations of patients over four years, with each patient assigned to a medical practice, and treatment assigned at the practice level. We make several simplifications to the problem: we restrict to a single year of data and work with a subset of 20\% of the data for computational reasons. We treat the practice to which a patient is assigned as a covariate, and set the probability that a patient is treated to be equal to the probability that their practice is treated. Finally, we do not account for the cluster structure of the data in our analysis \citep{abadie2017}; sample trimming in the presence of cluster structure is an exciting direction for future work. After these simplifications, there are 65,609 patients in the sample, whose propensities range from 0.11 to 0.85, suggesting that there are not significant overlap issues in the data. Meanwhile, responses range from $-\$7.20$ (indicating a patient paid Medicare, rather than the other way around) to $\$67,847$, suggesting there may potentially be significant heteroscedasticity. 

For our analysis, we use the simultaneous trimming method of Section~\ref{sec:simul}. Specifically, we use heteroscedastic fixed-fraction trimming with fractions $\{0, 0.1, 0.2, 0.3\}$; these methodological choices are based on our prior belief that there is a small fraction of patients with extreme responses, and that trimming these patients may significantly reduce variance. (For reference, the optimal heteroscedastic trimming trims nearly 90\% of the data, and is thus not an appealing option in this case.) We fit propensity scores using practice-level covariate and conditional means and variances using patient-level covariates. As in the experiments of Section~\ref{subsec:nhanes}, based on the discussion following Assumption~\ref{ass:trim}, all regressions are fit using the $\texttt{regression\_forest}$ function from the R package $\texttt{grf}$ \citep{grf}, \revision{and the regressions are cross-fit when selecting which units to trim, and not cross-fit when performing inference.}

The results of this analysis are reported in Table~\ref{tab:acic} and visualized in Figure~\ref{fig:acic}. These results paint a clear picture of the effects of the intervention: we see from the left panel of Figure~\ref{fig:acic} that there is significant heterogeneity with the treatment increasing costs for some patients, and decreasing costs for others, and that there are a handful of outliers with large AIPW scores that significantly affect our estimate. In the right panel of the same figure, we find that the treatment effect is slightly positive on the full population, but with a wide confidence interval due to the variance introduced by units with large responses. However, after trimming a small fraction of the data, the treatment effect is negative and the confidence interval around it is much narrower, because we have trimmed off the extreme units from the left panel. Furthermore, for all three sub-populations, we can conclude that this negative effect is statistically significant, in part because of the much lower variance of these populations. The change in point estimate also contributes to significance, but we note that if the sub-populations had the same standard error as the full population, their confidence intervals would contain zero. Thus the combination of simultaneous valid confidence intervals and heteroscedastic trimming effectively served to remove the handful of units distorting our estimate of the treatment effect and allowed us to identify a large sub-population of the data on which the treatment has a statistically significant negative effect.

Crucially, we would also not be able to draw this conclusion using classical sample trimming methods, since we would not know in advance which size of sub-population would be of interest. If we chose this sub-population adaptively from the data, the resulting inference would not be valid. By using our domain knowledge of the heavy-tailed behavior of the response and pre-committing to a set of trimming fractions, we are able to elucidate the structure of the response and draw statistically valid inferences.

The structure of this example---where the treatment effect on the full population is not significant but the treatment effect on a sub-population is---is illustrative, but not necessarily representative of all settings. In particular, it may well be the case that the treatment effect is not significant for any of the sub-populations considered, or significant for all of them. However, the goal of our method is not to find a sub-population on which there is a significant effect, but rather to find sub-populations with smaller confidence intervals than the full population, and allow for an analyst to choose freely between these and the full population. In this way, regardless of issues of statistical significance, the method can provide insight on the structure of treatment effects in the presence of outliers, and provide an analyst with statistically valid statements about potentially interesting sub-populations.

\section{Interpretive issues}
\revision{
	We have seen thus far that our theory allows us to perform inference around the estimand ${\tau}_{\hat{A}}$ using a doubly-robust estimator, that the resulting intervals have good coverage properties in simulations, and that these intervals are smaller than intervals provided by other trimming methods on real data. We now turn to the issue of interpreting $\hat{\tau}_{\hat{A}}$, focusing separately on the issues arising from modeling heteroscedasticity and the issues arising from the performing inference around ${\tau}_{\hat{A}}$ rather than $\tau_{\bar{A}}$. 

	Propensity-based trimming leads to an estimand that can be interpreted as the treatment effect for units that are sufficiently likely to receive both treatment and control, sometimes referred to as the population of clinical equipoise. In contrast, the heteroscedastically-trimmed estimand involves the interaction of two processes (the propensity score and the conditional variances), and is thus more difficult to interpret. However, we can still approximately interpret it by interpreting each of the two processes separately. From this perspective, it is the treatment effect for units that are sufficiently likely to receive both treatment and control, and also do not have highly atypical response distributions. (Of course, there may be some cancellation between the two processes here, which is why this interpretation is only an approximate one.) 

	Another potential interpretive issue is that $\tau_{\hat{A}}$ depends on quantities estimated from the data, and is not purely a functional of the underlying data distribution $\mathbb{P}$. For this reason, the estimate $\hat{\tau}_{\hat{A}}$ is best understood as the treatment effect on a specific subset of the present sample at hand\textemdash questions of whether that effect generalizes outside of the present sample must be answered by analyzing covariate distributions in the identified sample (as demonstrated in the NHANES data example in Section~\ref{subsec:nhanes}) and then comparing with covariate distributions in potential external populations. Domain knowledge and mechanistic understanding of the treatment are also helpful for this task.

	To better understand the relationship between $\tau_{\hat{A}}$ and $\tau_{\bar{A}}$, it is useful to study the difference when subtracting \eqref{eq:trimmed_clt} of Theorem~\ref{thm:trimmed_clt} from \eqref{eq:limit_trimmed_clt} of Theorem~\ref{thm:limit_trimmed_clt}, which gives $$\tau_{\hat{A}}-\tau_{\bar{A}}=\Delta_n+o_P(n^{-1/2}).$$ The correction term $\Delta_n$ is itself $o_P(1)$ (this follows from Lemma~\ref{lem:indicator_converge} in Appendix~\ref{sec:proofs}) and so we see that $\tau_{\hat{A}}\xrightarrow{\P} \tau_{\bar{A}}$ as $n\to \infty$. This means that the discrepancy between $\tau_{\hat{A}}$ and $\tau_{\bar{A}}$, which stems from finite-sample error in the estimation of nuisance components, vanishes in large samples; as such, the external validity of inferences around $\tau_{\hat{A}}$ may be more credible in larger sample sizes. 

	Despite these interpretive issues, we continue to recommend both modeling heteroscedasticity and performing inference around $\tau_{\hat{A}}$. Modeling heteroscedasticity provides, as we have seen in the experiment of Section~\ref{subsec:acic}, potentially much more variance reduction than assuming homoscedasticity, and can be thought of as, in the words of \citet{li2018}, ``defin[ing] and answer[ing] the question that can best be answered non-parametrically by the data at hand.'' With this perspective in mind, it is clear that modeling heteroscedasticity is a natural next step broadly in the literature on alternative, lower variance, estimands for problems with limited overlap. 

As for inference around $\tau_{\hat{A}}$, this approach allows us to perform inference around a trimmed estimand without smoothing the estimand, which further complicates interpretation. It is an endemic but underdiscussed challenge in the trimming literature that we cannot have our cake (an unsmoothed estimand) and eat it too (perform $\sqrt{n}$-inference) without accepting a model-dependent estimand like $\tau_{\hat{A}}$.}

\section{Conclusion}

In this paper we have extended the literature on trimming methods in two ways. First, we have proposed new methods for sample trimming that account for not only extreme propensities but also extreme conditional variances. Second, we have developed new theoretical results on inference after sample trimming, including a doubly-robust estimator and a bootstrap-based approach that gives simultaneously valid confidence intervals over multiple sub-populations. Our experiments show that our proposed estimators and confidence intervals achieve good coverage, reduce variance even when propensity-based trimming alone cannot, and enable analysts to choose between several different sub-populations based on any criteria they choose. This work raises several interesting future directions; we present two below.

\paragraph{Connections to heterogeneous treatment effects}

Some of the methods presented in this work may be remnisicient of the literature on estimating heterogeneous treatment effects. For example, \citet{athey2016} develop a method for identifying sub-populations of the data on which the treatment has a similar effect, and this is very nearly what was done in our experiments on the American Causal Inference Challenge data in Section~\ref{subsec:acic}. There, we found a sub-population of 90\% of the data for which our estimate of the treatment effect had much lower variance, but observed it was also one for which the point estimate of the treatment effect was substantially different. Thus it seems that we have identified 10\% of the data on which the treatment has a very different effect than on the other 90\%, which is essentially amounts to detecting a heterogeneous treatment effect. Such connections between heterogeneity in treatment effects and sample trimming have also been previously noted by \citet{crump2008}.

However, this interpretation is somewhat misleading: the trimming methods we have presented in this paper are focused on variance minimization, and if they also identify sub-populations on which point estimates are different, this is only because of the empirical fact that units with extreme variances often also have extreme conditional average treatment effects, not because our methods explicitly target treatment effect heterogeneity. The natural follow-up, then, is to explore variants of our methods that do incorporate information about treatment effect heterogeneity. For example, we might be interested in partitioning the sample into sub-populations on which the treatment has a similar effect, but also making sure that the variance of the point estimate on these sub-populations is not unacceptably high. Alternatively, we may consider forming sub-populations on which the treatment has a similar effect after trimming some fraction of outlier units on which the treatment has an unusually large or small effect, effectively creating a trimmed bin of ``other'' units in a partition of heterogeneous effects within the sample.

\paragraph{Inference around the original estimand} Finally, it is also worth revisiting the basic foundations of the trimming literature. The framework we have taken here, of performing inference around a modified estimand, is standard in the literature, but may not be ideal from an analyst's perspective, who may be unsure of which estimand is actually relevant to them. This is especially pertinent to our simultaneous trimming methods, since the path from a set of simultaneous confidence intervals around multiple estimands to a binary decision on implementing a policy is unclear. These issues are closely related to those that arise in problems with heterogeneous treatment effects, where a treatment may help some parts of the population but harm others, making the decision to treat or not fraught with complications. For these reasons, methods that do not shift the estimand but instead bound the impact of the trimmed units may be especially useful, and are worthy of further study.

\subsection*{Acknowledgements}
This work was supported in part by NSF CAREER Award \#2143176. We are grateful to Kevin Guo, Dominik Rothenhausler, and Art Owen for helpful comments on this work.

\bibliographystyle{plainnat}
\bibliography{trimming}

\newpage 

\appendix

\section{Proofs}
\label{sec:proofs}

In our proofs, as was also mentioned in Section~\ref{sec:model}, we employ empirical process notation~\citep{wellner2013, kennedy2016}. To summarize, we let $W_i = (X_i, Y_i, Z_i)$ be the entire triplet we observe for unit $i$, and we write $\P_nf=\frac{1}{n}\sum f(W_i)$ and $\P f= \int f(w)d\P(w)$. Note that for a random function $\hat{f}$, $\P \hat{f}$ is a random variable, since we do not integrate over the randomness in $\hat{f}$. In contrast, $\E[\hat{f}]$ is a deterministic quantity that integrates out the randomness in a new sample and in $\hat{f}$. We also define the norm $\| f\|_{L^q(\P)}=(\P |f|^q)^{1/q}$.

\subsection{Proof of Theorem~\ref{thm:trimmed_clt}}

In this section we prove Theorem~\ref{thm:trimmed_clt} from the main text. Before proceeding with the main proof, we introduce some notation and helpful lemmas. 

For the notation, we define $$\psi_1(W; \hat{\mu}_{(w)}, \hat{k}, \hat{\gamma})=(\hat{\mu}_1(X_i)-\hat{\mu}_0(X_i)-\tau(X_i))\mathbf{1}\{\hat{k}(X_i)\leq \hat{\gamma}\}$$ and $$\psi_2(W; \hat{\mu}_{(w)}, \hat{e}, \hat{k}, \hat{\gamma})=\left(\frac{Y_i(1)-\hat{\mu}_1(X_i)}{\hat{e}(X_i)}Z_i+\frac{Y_i(0)-\hat{\mu}_0(X_i)}{1-\hat{e}(X_i)}\right)\mathbf{1}\{\hat{k}(X_i)\leq \hat{\gamma}\},$$ so that the statement of the theorem is that 
\begin{equation}
	\hat{\tau}_{\hat{A}}-\tau_{\hat{A}}=\frac{1}{\P(\bar{k}(X_i)\leq \bar{\gamma})}\P_n\psi_2(W; \mu_{(w)}, e, \bar{k}, \bar{\gamma})+o_P(n^{-1/2}).
	\label{eq:trimmed_clt_app}
\end{equation}

Next, we show through two lemmas that the convergence assumptions on $\hat{\mu}_{(w)}$, $\hat{e}$, and $\hat{k}$ imply several other convergences that we will require. 

\begin{lemma}
	\label{lem:indicator_converge}
	Suppose $\hat{k}$ and $\hat{\gamma}$ satisfy Assumption~\ref{ass:trim}. Then 
	\begin{equation}
		\|\mathbf{1}\{\hat{k}(X_i)\leq \hat{\gamma}\}-\mathbf{1}\{\bar{k}(X_i)\leq \bar{\gamma}\}\|_{L^2(\P)}\xrightarrow{\P}0.
		\label{eq:indicator_converge}
	\end{equation}
\end{lemma}
\begin{proof}
	We begin by writing 
		\begin{multline}
			\|\mathbf{1}\{\hat{k}(X_i)\leq \hat{\gamma}\}-\mathbf{1}\{\bar{k}(X_i)\leq \bar{\gamma}\}\|_{L^2(\P)}\leq \\ \|\mathbf{1}\{\hat{k}(X_i)\leq \hat{\gamma}\}-\mathbf{1}\{\hat{k}(X_i)\leq \bar{\gamma}\}\|_{L^2(\P)}+\|\mathbf{1}\{\hat{k}(X_i)\leq \bar{\gamma}\}-\mathbf{1}\{\bar{k}(X_i)\leq \bar{\gamma}\}\|_{L^2(\P)},
		\label{eq:triangle}
	\end{multline}
	by the triangle inequality, and then bound each of the two terms on the RHS of \eqref{eq:triangle} separately. 

	For the first term on the RHS of \eqref{eq:triangle}, we have 
	\begin{align}
		\|\mathbf{1}\{\hat{k}(X_i)\leq \hat{\gamma}\}-\mathbf{1}\{\hat{k}(X_i)\leq \bar{\gamma}\}\|_{L^2(\P)}^2&=\int \left(\mathbf{1}\{\hat{k}(x)\leq \hat{\gamma}\}-\mathbf{1}\{\hat{k}(x)\leq \bar{\gamma}\}\right)^2 d\P(x),\\
		&=\int \mathbf{1}\{\min(\hat{\gamma}, \bar{\gamma})\leq \hat{k}(x)\leq \max(\hat{\gamma}, \bar{\gamma})\} d\P(x),\\
			&\leq B|\hat{\gamma}-\gamma|,\\
			&=o_P(1),\label{eq:first_op1}
	\end{align}
	where $B$ is an upper bound on the density of $\hat{k}(x)$ from Assumption~\ref{ass:trim}, and the final equality uses the fact that $\hat{\gamma}$ converges to $\gamma$ in probability.

	For the second term on the RHS of \eqref{eq:triangle}, we assume without the loss of generality that $\hat{k}(x)\leq \bar{k}(x)$ for all $x$. Then, for any $\epsilon>0$, 

	\begin{align}
		&\|\mathbf{1}\{\hat{k}(x)< \bar{\gamma}\}-\mathbf{1}\{\bar{k}(x)< \bar{\gamma}\}\|_{L^2(\P)},\\
		=& \int \left( \mathbf{1}\{\hat{k}(x)< \bar{\gamma}\}-\mathbf{1}\{\bar{k}(x)< \bar{\gamma}\} \right)^2\, d\P(x),\\
		=& \int \mathbf{1}\{\hat{k}(x)< \bar{\gamma}< \bar{k}(x)\}\, d\P(x),\\
		=& \int \mathbf{1}\{\hat{k}(x)< \bar{\gamma}< \bar{k}(x)\}\mathbf{1}\{|\bar{k}(x)-\bar{\gamma}|\geq \epsilon\}+\mathbf{1}\{\hat{k}(x)\ \bar{\gamma}\leq \bar{k}(x)\}\mathbf{1}\{ |\bar{k}(x)-\bar{\gamma}|<\epsilon\}d\P(x),\\
		\leq&\int \mathbf{1}\{\hat{k}(x)\leq \bar{\gamma}\leq \bar{k}(x)\}\mathbf{1}\{|\bar{k}(x)-\bar{\gamma}|\geq \epsilon\}d\P(x)+\int\mathbf{1}\{|\bar{k}(x)-\bar{\gamma}|<\epsilon\} d\P(x),\\
		\leq &\int \mathbf{1}\{|\hat{k}(x)-\bar{k}(x)|\geq \epsilon\}\, d\P(x)+\P(|\bar{k}(X_i)-\bar{\gamma}|<\epsilon),\\
		\leq &\int \mathbf{1}\{|\hat{k}(x)-\bar{k}(x)|\geq \epsilon\}\, d\P(x)+2B\epsilon, \label{eq:k_hat_minus_k}
	\end{align}
\end{proof}
since $\bar{k}$ has density bounded by $B$.

We will now argue that the first term of \eqref{eq:k_hat_minus_k} is $o_P(1)$ under our assumptions. Indeed, for any $\epsilon'>0$, 
\begin{align}
	\P\left( \int \mathbf{1}\{|\hat{k}(x)-\bar{k}(x)|\geq \epsilon\}\, d\P(x)>\epsilon' \right)&\leq \P\left( \int \mathbf{1}\{|\hat{k}(x)-\bar{k}(x)|\geq \epsilon\}\, d\P(x)>0 \right),\\
	&\leq \P\left( \max_x |\hat{k}(x)- \bar{k}(x)|>\epsilon \right),\\
	&= o(1),
\end{align}
by Assumption~\ref{ass:trim}. Thus, the second term on the RHS of \eqref{eq:triangle} is $o_P(1)+2B\epsilon$ for any $\epsilon>0$, and this implies that it is $o_P(1)$ it self. Thus both terms of \eqref{eq:triangle} are $o_P(1)$, and the result follows. 

\begin{lemma}
	\label{lem:psi_converge}
	Suppose $\hat{\mu}_{(w)}$, $\hat{e}$, $\hat{k}$, and $\hat{\gamma}$ satisfy Assumptions. Then,
	\begin{equation}
		\| \psi_1(W; \hat{\mu}_{(w)}, \hat{k}, \hat{\gamma}) - \psi_1(W; \mu_{(w)}, k, \gamma)\|_{L^2(\P)}\xrightarrow{\P} 0
		\label{eq:psi1_convergence}
	\end{equation}
	and 
\begin{equation}
	\| \psi_2(W; \hat{\mu}_{(w)}, \hat{e}, \hat{k}, \hat{\gamma}) - \psi_2(W; \mu_{(w)}, e, k, \gamma)\|_{L^2(\P)}\xrightarrow{\P} 0.
	\label{eq:psi2_convergence}
	\end{equation}
\end{lemma}

\begin{proof}
	We show only \eqref{eq:psi2_convergence}, since the proof for \eqref{eq:psi1_convergence} is similar. By the triangle inequality, $$\| \psi_2(W; \hat{\mu}_{(w)}, \hat{e}, \hat{k}, \hat{\gamma}) - \psi_2(W; \mu_{(w)}, e, k, \gamma)\|_{L^2(\P)}$$ is bounded by 
	\begin{multline}
		\left\| \left( \frac{Y_i(1)-\hat{\mu}_{(1)}(X_i)}{\hat{e}(X_i)}Z_i\mathbf{1}\{\hat{k}(X_i)\leq \hat{\gamma}\}-\frac{Y_i(1)-\mu_{(1)}(X_i)}{e(X_i)}Z_i\mathbf{1}\{\bar{k}(X_i)\leq \bar{\gamma}\}\right)\right\|_{L^2(\P)}\\
		+\left\| \left( \frac{Y_i(0)-\hat{\mu}_{(0)}(X_i)}{1-\hat{e}(X_i)}(1-Z_i)\mathbf{1}\{\hat{k}(X_i)\leq \hat{\gamma}\}-\frac{Y_i(0)-\mu_{(0)}(X_i)}{1-e(X_i)}(1-Z_i)\mathbf{1}\{\bar{k}(X_i)\leq \bar{\gamma}\}\right)\right\|_{L^2(\P)}.
		\label{eq:psi_split}
\end{multline}
We will show that the first term of \eqref{eq:psi_split} is $o_P(1)$; the second term can be handled similarly. 

The first term of \eqref{eq:psi_split} is, by the triangle inequality, bounded by
\begin{multline}
		\left\| \left( \frac{Y_i(1)-\hat{\mu}_{(1)}(X_i)}{\hat{e}(X_i)}Z_i\mathbf{1}\{\hat{k}(X_i)\leq \hat{\gamma}\}-\frac{Y_i(1)-\mu_{(1)}(X_i)}{e(X_i)}Z_i\mathbf{1}\{\hat{k}(X_i)\leq \hat{\gamma}\}\right)\right\|_{L^2(\P)}\\
		+\left\| \left( \frac{Y_i(1)-{\mu}_{(1)}(X_i)}{{e}(X_i)}Z_i\mathbf{1}\{\hat{k}(X_i)\leq \hat{\gamma}\}-\frac{Y_i(1)-\mu_{(1)}(X_i)}{e(X_i)}Z_i\mathbf{1}\{\bar{k}(X_i)\leq \bar{\gamma}\}\right)\right\|_{L^2(\P)}\\
		\label{eq:psi_split_2}
\end{multline}
Again, we will show that the first term of \eqref{eq:psi_split_2} is $o_P(1)$ and the other term can be handled similarly. 

The first term of \eqref{eq:psi_split_2} is 
\begin{align}
	&\left\| \left( \frac{Y_i(1)-\hat{\mu}_{(1)}(X_i)}{\hat{e}(X_i)}Z_i\mathbf{1}\{\hat{k}(X_i)\leq \hat{\gamma}\}-\frac{Y_i(1)-\mu_{(1)}(X_i)}{e(X_i)}Z_i\mathbf{1}\{\hat{k}(X_i)\leq \hat{\gamma}\}\right)\right\|_{L^2(\P)},\\
	&=\left\| \left( \frac{{\mu}_{(1)}(X_i)-\hat{\mu}_{(1)}(X_i)}{\hat{e}(X_i)}Z_i+(Y_i(1)-\mu_{(1)}(X_i))Z_i\left( \frac{1}{\hat{e}(X_i)}-\frac{1}{e(X_i)} \right)\right)\mathbf{1}\{\hat{k}(X_i)\leq \hat{\gamma}\} \right\|_{L^2(\P)},\\
	&\leq \left \|\frac{\mu_{(1)}(X_i)-\hat{\mu}_{(1)}(X_i)}{\hat{e}(X_i)}Z_i\mathbf{1}\{\hat{k}(X_i)\leq \hat{\gamma}\}\right\|_{L^2(\P)}+\left\|(Y_i(1)-\mu_{(1)}(X_i))Z_i\left( \frac{1}{\hat{e}(X_i)}-\frac{1}{e(X_i)} \right)\mathbf{1}\{\hat{k}(X_i)\leq \hat{\gamma}\}\right\|_{L^2(\P)}.
	\label{eq:decomp}
\end{align}
We now control each term of \eqref{eq:decomp} separately. First,
\begin{align}
	\left \|\frac{\mu_{(1)}(X_i)-\hat{\mu}_{(1)}(X_i)}{\hat{e}(X_i)}Z_i\mathbf{1}\{\hat{k}(X_i)\leq \hat{\gamma}\}\right\|_{L^2(\P)}^2&\lesssim \left\| (\mu_{(1)}(X_i)-\hat{\mu}_{(1)}(X_i))\mathbf{1}\{\hat{k}(X_i)\leq \hat{\gamma}\}\right\|_{L^2(\P)},\\
	&=o_P(1),
	\label{eq:first_factor}
\end{align}
where the first inequality is because $Z_i\in \{0,1\}$ and $\hat{e}(X_i)$ is bounded away from 0 and 1 as $n\to \infty$ by consistency, and the second is by Assumption~\ref{ass:dml}.

The second term of \eqref{eq:decomp} requires a truncation argument. Observe that for any $M>0$,
\begin{align}
	&\left\|(Y_i(1)-\mu_{(1)}(X_i))Z_i\left( \frac{1}{\hat{e}(X_i)}-\frac{1}{e(X_i)} \right)\mathbf{1}\{\hat{k}(X_i)\leq \hat{\gamma}\}\right\|_{L^2(\P)}\\ 
	= &\left\|(Y_i(1)-\mu_{(1)}(X_i))\mathbf{1}\{|Y_i(1)-\mu_{(1)}(X_i)|>M\}Z_i\left( \frac{1}{\hat{e}(X_i)}-\frac{1}{e(X_i)}\right)\mathbf{1}\{\hat{k}(X_i)\leq \hat{\gamma}\} \right\|_{L^2(\P)},\\
	&+\left\|(Y_i(1)-\mu_{(1)}(X_i))\mathbf{1}\{|Y_i(1)-\mu_{(1)}(X_i)|\leq M\}Z_i\left( \frac{1}{\hat{e}(X_i)}-\frac{1}{e(X_i)}\right)\mathbf{1}\{\hat{k}(X_i)\leq \hat{\gamma}\} \right\|_{L^2(\P)},\\
	\lesssim &\left\|(Y_i(1)-\mu_{(1)}(X_i))\mathbf{1}\{|Y_i(1)-\mu_{(1)}(X_i)|>M\}\right\|_{L^2(\P)}+\left\|MZ_i\left( \frac{1}{\hat{e}(X_i)}-\frac{1}{e(X_i)}\right)\mathbf{1}\{\hat{k}(X_i)\leq \hat{\gamma}\} \right\|_{L^2(\P)},\\
	= &\left\|(Y_i(1)-\mu_{(1)}(X_i))\mathbf{1}\{|Y_i(1)-\mu_{(1)}(X_i)|>M\}\right\|_{L^2(\P)}+o_P(1)
	\label{eq:truncated}
\end{align}
where the second inequality again uses the fact that $Z_i\mathbf{1}\{\hat{k}(X_i)\leq \hat{\gamma}\}\in \{0,1\}$ and that $\frac{1}{\hat{e}}-\frac{1}{e}$ is bounded away from 0 and 1 as $n\to \infty$ by consistency and the third equality uses Assumption~\ref{ass:dml}.  Now, as $M\to \infty$, the first term of \eqref{eq:truncated} vanishes because $Y_i(1)$ has finite variance, so we conclude that the second term of \eqref{eq:decomp} is in fact $o_P(1)$. 

Having shown that both terms of \eqref{eq:decomp} are $o_P(1)$, we find that \eqref{eq:decomp} itself is $o_P(1)$ as well, completing the proof. 
\end{proof}

With these lemmas in hand, we can now prove Theorem~\ref{thm:trimmed_clt}. 
	\begin{proof}[Proof of Theorem~\ref{thm:trimmed_clt}]
		In the notation introduced above, we have
	\begin{equation}
		\hat{\tau}_{\hat{A}}-\tau_{\hat{A}}=\frac{\P_n\left(\psi_1(W; \hat{\mu}_{(w)}, \hat{k}, \hat{\gamma})+\psi_2(\hat{\mu}_{(w)}, \hat{e}, \hat{k}, \hat{\gamma})\right)}{n_{\hat{A}}/n}.
		\label{eq:notation}
	\end{equation}
	We now proceed in two steps. The first step is to show that 
	\begin{equation}
		\frac{n_{\hat{A}}}{n}\xrightarrow{\P} \P(\bar{k}(X_i)\leq \bar{\gamma})
	\label{eq:set_size}
	\end{equation}
	and the second is to show that
	\begin{equation}
		\P_n\left( \psi_1(W; \hat{\mu}_{(w)}, \hat{k}, \hat{\gamma})+\psi_2(\hat{\mu}_{(w)}, \hat{e}, \hat{k}, \hat{\gamma}) \right)=\P_n\psi_2(W; \mu_{(w)}, e, \bar{k} ,\bar{\gamma})+o_P(n^{-1/2}).
		\label{eq:tau_hat_minus_tau}
	\end{equation}
	Once we establish \eqref{eq:set_size} and \eqref{eq:tau_hat_minus_tau}, the desired result \eqref{eq:trimmed_clt_app} follows from Slutsky's lemma.

	For the first step, we begin by noting that $n_{\hat{A}}/n=\P_n\mathbf{1}\{\hat{k}(X_i)\leq \hat{\gamma}\}$ and then computing 
	\begin{align}
		\P_n\mathbf{1}\{\hat{k}(X_i)\leq \hat{\gamma}\}&=(\P_n-\P)\mathbf{1}\{\hat{k}(X_i)\leq \hat{\gamma}\}+\P\mathbf{1}\{\hat{k}(X_i)\leq \hat{\gamma}\},\\
		&= (\P_n-\P)\mathbf{1}\{\bar{k}(X_i)\leq \bar{\gamma}\}+\P\mathbf\{\hat{k}(X_i)\leq \hat{\gamma}\}+o_P(n^{-1/2}),\\
		&= \P_n\mathbf{1}\{\bar{k}(X_i)\leq \bar{\gamma}\}+\P(\mathbf{1}\{\hat{k}(X_i)\leq \hat{\gamma}\}-\mathbf{1}\{\bar{k}(X_i)\leq \bar{\gamma}\})+o_P(n^{-1/2}),\label{eq:denom_diff},\\
		&\leq \P_n\mathbf{1}\{\bar{k}(X_i)\leq \bar{\gamma}\}+\|\mathbf{1}\{\hat{k}(X_i)\leq \hat{\gamma}\}-\mathbf{1}\{\bar{k}(X_i)\leq \bar{\gamma}\}\|_{L^1(\P)}+o_P(n^{-1/2}),\\
		&= \P_n\mathbf{1}\{\bar{k}(X_i)\leq \bar{\gamma}\}+o_P(1),\\
		&= \P(\bar{k}(X_i)\leq \bar{\gamma})+o_P(1),
		\label{eq:denom}
	\end{align}
	where the second equality uses either Lemma 2 of \citet{kennedy2020} and Assumption~\ref{ass:cross}(a) or Lemma 19.24 of \citet{van2000}, Assumption~\ref{ass:cross}(b), and the fact that $\mathbf{1}\{\hat{k}(X_i)\leq \hat{\gamma}\}$ is Donsker because $\hat{k}$ lies in a VC-subgraph class and the indicator is a monotone function \citep{wellner2013}. The fifth equality uses Lemma~\ref{lem:indicator_converge} and the fact that $L^2$-convergence implies $L^1$-convergence, and the sixth equality uses the law of large numbers. Since \eqref{eq:denom} implies \eqref{eq:set_size}, we are done with the first step. 

	We now proceed to the second step of showing \eqref{eq:tau_hat_minus_tau}. Note that, by arguments similar to those in the previous paragraph, $\psi_1$ and $\psi_2$ are both Donsker under Assumption~\ref{ass:trim}(b).

	With this in mind, we first consider $\P_n\psi_1$ and find that 
	\begin{align}
		\P_n \psi_1(W; \hat{\mu}_{(w)}, \hat{k}, \hat{\gamma})&=(\P_n-\P)\psi_1(W; \hat{\mu}_{(w)}, \hat{k}, \hat{\gamma})+\P \psi_1(W; \hat{\mu}_{(w)}, \hat{k}, \hat{\gamma}),\\
		&= (\P_n-\P)\psi_1(W; \mu_{(w)}, \hat{k}, \hat{\gamma})+\P \psi_1(W; \hat{\mu}_{(w)}, \hat{e}, \hat{k}, \hat{\gamma})+o_P(n^{-1/2})\label{eq:hats_off_1},\\
		&= \P \psi_1(W; \hat{\mu}_{(w)}, \hat{k}, \hat{\gamma})+o_P(n^{-1/2}),
		\label{eq:m1}
	\end{align}
	where the second equality uses either Assumption~\ref{ass:cross}(a) and Lemma 2 of \citet{kennedy2020} or Assumption~\ref{ass:cross}(b), Lemma~\ref{lem:psi_converge}, Lemma 19.24 of \citet{van2000}, and the fact that $\psi_1$ and $\psi_2$ are Donsker. The third equality follows from the fact that $\psi_1(W; \mu_{(w)}, \hat{k}, \hat{\gamma})=0$ identically. 

	Similarly,
	\begin{align}
		\P_n \psi_2(W; \hat{\mu}_{(w)}, \hat{e}, \hat{k}, \hat{\gamma})&= (\P_n-\P) \psi_2(W; \hat{\mu}_{(w)}, \hat{e}, \hat{k}, \hat{\gamma})+\P \psi_2(W; \hat{\mu}_{(w)}, \hat{e}, \hat{k}, \hat{\gamma}),\\
		&= (\P_n-\P) \psi_2(W; \mu_{(w)}, e, k, \gamma)+\P \psi_2(W;\hat{\mu}_{(w)}, \hat{e}, \hat{k}, \hat{\gamma})+o_P(n^{-1/2}),\label{eq:hats_off}\\
		&= \P_n \psi_2(W; \mu_{(w)}, e, k, \gamma)+\P \psi_2(W; \hat{\mu}_{(w)}, \hat{e}, \hat{k}, \hat{\gamma})+o_P(n^{-1/2}),
		\label{eq:m2}
	\end{align}
	where \eqref{eq:hats_off} is justified in the same way as \eqref{eq:hats_off_1} and the third equality is justified by the fact that $\P \psi_2(W; \mu_{(w)}, e, k, \gamma)=0$ identically. 

	Taken together, \eqref{eq:m1} and \eqref{eq:m2} imply that 
	\begin{equation}
		\P_n(\psi_1+\psi_2)=\P_n\psi_2(W; \mu_{(w)}, e, k, \gamma )+\P\left(\psi_1(W; \hat{\mu}_{(w)}, \hat{e}, \hat{k}, \hat{\gamma})+\psi_2(W; \hat{\mu}_{(w)}, \hat{e}, \hat{k}, \hat{\gamma})\right)+o_P(n^{-1/2}).
		\label{eq:sum_m1_m2}
	\end{equation}
Now, the second term of \eqref{eq:sum_m1_m2} is 
\begin{align}
	&= \P\left[ \left( \hat{\mu}_1(X_i)-\hat{\mu}_0(X_i)-\tau(X_i)+\frac{Y_i-\hat{\mu}_1(X_i)}{\hat{e}(X_i)}Z_i+\frac{Y_i-\hat{\mu}_0(X_i)}{1-\hat{e}(X_i)}(1-Z_i) \right)\mathbf{1}\{\hat{k}(X_i)\leq \hat{\gamma}\} \right],\\
	&= \sum_{w\in \{0,1\}}\P\left[ \left( \hat{\mu}_{(w)}(X_i)-\mu_{(w)}(X_i)+\frac{Y_i-\hat{\mu}_{(w)}(X_i)}{\hat{e}(X_i)}Z_i \right)\mathbf{1}\{\hat{k}(X_i)\leq \hat{\gamma}\} \right].
\end{align}

We bound the $w=1$ term; the $w=0$ term is analogous. The $w=1$ term is, after applying the tower rule conditional on $X_i$ and simplifying,
\begin{align}
	&= \P\left[ \left( \frac{(\hat{\mu}_1(X_i)-\mu_1(X_i))\hat{e}(X_i)+(\mu_1(X_i)-\hat{\mu}_1(X_i))e(X_i)}{\hat{e}(X_i)}\mathbf{1}\{\hat{k}(X_i)\leq \hat{\gamma}\} \right) \right],\\
	&\lesssim \P\left[ (\hat{\mu}_1(X_i)-\mu_1(X_i))(\hat{e}(X_i)-e(X_i))\mathbf{1}\{\hat{k}(X_i)\leq \hat{\gamma}\} \right],\\
	&\leq \sqrt{\P\left((\hat{\mu}_1(X_i)-\mu_1(X_i))^2\mathbf{1}\{\hat{k}(X_i)\leq \hat{\gamma}\} \right)\cdot \P\left( (\hat{e}(X_i)-e(X_i))^2 \mathbf{1}\{\hat{k}(X_i)\leq \hat{\gamma}\} \right)},\\
	&= \| (\hat{\mu}_1(X_i)-\mu_1(X_i))\mathbf{1}\{\hat{k}(X_i)\leq \hat{\gamma}\}\|_{L^2(\P)}\times \|(\hat{e}(X_i)-e(X_i))\mathbf{1}\{\hat{k}(X_i)\leq \hat{\gamma}\}\|_{L^2(\P)},
	\label{eq:norm_product}
\end{align}

where the first inequality uses the fact that $\hat{e}(X_i)$ is bounded away from $0$ and $1$ for sufficiently large $n$ by consistency and the second uses Cauchy-Schwarz. Then, \eqref{eq:norm_product} is $o_P(n^{-1/2})$ by Assumption~\ref{ass:dml}, which establishes \eqref{eq:tau_hat_minus_tau} and completes the proof.
\end{proof}

\subsection{Proof of Theorem~\ref{thm:limit_trimmed_clt}}

In this section, we analyze $\hat{\tau}_{\hat{A}}-\tau_{\bar{A}}$ and show that 
\begin{equation}
		\hat{\tau}_{\hat{A}}-\tau_{\bar{A}}=\frac{1}{\P(\bar{k}(X_i)\leq \bar{\gamma})}\left(\frac{1}{n}\sum_{i=1}^n \left(\frac{Y_i-\mu_1(X_i)}{e(X_i)}Z_i+\frac{Y_i(0)-\mu_0(X_i)}{1-e(X_i)}(1-Z_i)\right)\mathbf{1}\{\bar{k}(X_i)\leq\bar{\gamma}\}+\Delta_n\right),
		\label{eq:thm2_statement}
	\end{equation}
	for 
	\begin{equation}
		\Delta_n=\P\left[ \left(\tau(X_i)-\E[\tau(X_i)\mid \bar{k}(X_i)\leq \bar{\gamma}]\right)\left(\mathbf{1}\{\hat{k}(X_i)\leq \hat{\gamma}\}-\mathbf{1}\{\bar{k}(X_i)\leq \bar{\gamma}\}\right) \right]+o_P(n^{-1/2}),
		\label{eq:delta_def}
	\end{equation}
as claimed in Theorem~\ref{thm:limit_trimmed_clt}.

This is slightly more involved than the results of the previous section because both the numerators and the denominators of $\hat{\tau}_{\hat{A}}$ and $\tau_{\bar{A}}$ differ. Thus we we first analyze the numerators, then the denominators, and then combine our results to obtain the theorem. To simplify notation, let 
\begin{equation}
	\hat{S}_{\hat{A}}=\sum_{i=1}^n \hat{\tau}(X_i)\mathbf{1}\{\hat{k}(X_i)\leq \hat{\gamma}\},\quad S_{\bar{A}}=\sum_{i=1}^n \tau(X_i)\mathbf{1}\{\bar{k}(X_i)\leq \bar{\gamma}\},
\end{equation}
be the numerators of $\hat{\tau}_{\hat{A}}$ and $\tau_{\bar{A}}$, so that $\hat{\tau}_{\hat{A}}=\hat{S}_{\hat{A}}/n_{\hat{A}}$ and $\tau_{\bar{A}}=S_{\bar{A}}/n_{\bar{A}}$. 

The following lemma analyzes the difference between $\hat{S}_{\hat{A}}$ and $S_{\bar{A}}$. 

\begin{lemma}
	\label{lem:numerator}
	We have
	\begin{multline}
		\frac{\hat{S}_{\hat{A}}-S_{\bar{A}}}{n}= \frac{1}{n}\sum_{i=1}^n \left(\frac{Y_i-\mu_1(X_i)}{e(X_i)}Z_i+\frac{Y_i(0)-\mu_0(X_i)}{1-e(X_i)}(1-Z_i)\right)\mathbf{1}\{\bar{k}(X_i)\leq \bar{\gamma}\}\\ +\P[\tau(X_i)(\mathbf{1}\{\hat{k}(X_i)\leq \hat{\gamma}\}-\mathbf{1}\{\bar{k}(X_i)\leq \bar{\gamma}\})]+o_P(n^{-1/2}).
	\end{multline}
\end{lemma}
\begin{proof}
	As before, we begin by defining $$\psi(W; \hat{\mu}_{(w)}, \hat{e}, \hat{k}, \hat{\gamma})=\left(\hat{\mu}_1(X_i)-\hat{\mu}_0(X_i)+\frac{Y_i(1)-\hat{\mu}_1(X_i)}{\hat{e}(X_i)}Z_i-\frac{Y_i(0)-\hat{\mu}_0(X_i)}{1-\hat{e}(X_i)}(1-Z_i)\right)\mathbf{1}\{\hat{k}(X_i)\leq \hat{\gamma}\},$$ so that $\hat{S}_{\hat{A}}/n=\P_n \psi(W; \hat{\mu}_{(w)}, \hat{e}, \hat{k}, \hat{\gamma}).$ We can check that $\psi(W; \hat{\mu}_{(w)}, \hat{e}, \hat{k})$ converges in $L^2$ to $\psi(W; \mu_{(w)}, e, \bar{k})$ using arguments like those in the proof of Lemma~\ref{lem:psi_converge}, and so we can then write 
	\begin{align}
		\hat{S}_{\hat{A}}/n&=\P_n \psi(W; \hat{\mu}_{(w)}, \hat{e}, \hat{k}, \hat{\gamma}),\\
		&=(\P_n-\P)\psi(W; \hat{\mu}_{(w)}, \hat{e}, \hat{k}, \hat{\gamma})+\P\psi (W; \hat{\mu}_{(w)}, \hat{e}, \hat{k}, \hat{\gamma}),\\
		&=(\P_n-\P)\psi(W; {\mu}_{(w)}, {e}, \bar{k}, \bar{\gamma})+\P\psi (W; \hat{\mu}_{(w)}, \hat{e}, \hat{k}, \hat{\gamma})+o_P(n^{-1/2}),\\
		&= S_{\bar{A}}/n+\frac{1}{n}\sum_{i=1}^n\left( \frac{Y_i(1)-\mu_1(X_i)}{e(X_i)}-\frac{Y_i(0)-\mu_0(X_i)}{1-e(X_i)} \right)\mathbf{1}\{\bar{k}(X_i)\leq \bar{\gamma}\}\\ 
		&\quad +\P(\psi (W; \hat{\mu}_{(w)}, \hat{e}, \hat{k}, \hat{\gamma})-\psi(W; \mu_{(w)}, e, \bar{k}, \bar{\gamma}))+o_P(n^{-1/2})\nonumber ,
		\label{eq:psi}
	\end{align}
	where the third equality uses Assumption~\ref{ass:cross}(a) and Lemma 2 of \citet{kennedy2020} or Assumption~\ref{ass:cross}(b), Lemma 19.24 of \citet{van2000}, and the fact that $\psi$ converges and is Donsker. Now, rearranging \eqref{eq:psi} gives
	\begin{align}
		\frac{\hat{S}_{\hat{A}}-S_{\bar{A}}}{n}=&\frac{1}{n}\sum_{i=1}^n \left( \frac{Y_i(1)-\mu_1(X_i)}{e(X_i)}Z_i-\frac{Y_i(0)-\mu_0(X_i)}{e(X_i)}(1-Z_i) \right)\\ &+\P(\psi (W; \hat{\mu}_{(w)}, \hat{e}, \hat{k}, \hat{\gamma})-\psi(W; \mu_{(w)}, e, \bar{k}, \bar{\gamma}))+o_P(n^{-1/2})\nonumber .
		\label{eq:simplified}
	\end{align}

	All that is left to do is control the error term in \eqref{eq:simplified}. We analyze the terms corresponding to the treatment group; the control group terms are analogous. Those terms are 
	\begin{align}
		&= \P\left[ \left(\hat{\mu}_1(X_i)+\frac{Y_i(1)-\hat{\mu}_1(X_i)}{\hat{e}(X_i)}Z_i\right)\mathbf{1}\{\hat{k}(X_i)\leq \hat{\gamma}\}-\left( \mu_1(X_i)+\frac{Y_i(1)-\mu_1(X_i)}{e(X_i)}Z_i \right)\mathbf{1}\{\bar{k}(X_i)\leq \bar{\gamma}\}\right],\\
		&= \P\left[ \left( \frac{\hat{\mu}_1(X_i)\hat{e}(X_i)+(Y_i(1)-\hat{\mu}_1(X_i))e(X_i)}{\hat{e}(X_i)} \right)\mathbf{1}\{\hat{k}(X_i)\leq \hat{\gamma}\} -\mu_1(X_i)\mathbf{1}\{\bar{k}(X_i)\leq \bar{\gamma}\}\right],\\
		&= \P\left[ \frac{\hat{\mu}_1\hat{e}\mathbf{1}\{\hat{k}\leq \hat{\gamma}\}+\mu_1e\mathbf{1}\{\hat{k}\leq \hat{\gamma}\}-\hat{\mu}_1e\mathbf{1}\{\hat{k}\leq \hat{\gamma}\}-\mu_1\hat{e}\mathbf{1}\{\bar{k}\leq \bar{\gamma}\}}{\hat{e}} \right],
		\label{eq:r2}
	\end{align}
	where we omit the covariates $X_i$ in \eqref{eq:r2} for brevity. 

	Then, we can add and subtract $\mu_1\hat{e}\mathbf{1}\{\hat{k}\leq \hat{\gamma}\}$ from the numerator of \eqref{eq:r2} and factor it as
	\begin{align}
		& =\hat{\mu}_1\hat{e}\mathbf{1}\{\hat{k}\leq \hat{\gamma}\}+\mu_1e\mathbf{1}\{\hat{k}\leq \hat{\gamma}\}-\hat{\mu}_1e\mathbf{1}\{\hat{k}\leq \hat{\gamma}\}-\mu_1\hat{e}\mathbf{1}\{\hat{k}\leq \hat{\gamma}\} +\mu_1\hat{e}\mathbf{1}\{\hat{k}\leq \hat{\gamma}\}-\mu_1\hat{e}\mathbf{1}\{\bar{k}\leq \bar{\gamma}\},\\
		&= (\hat{\mu}_1-\mu_1)(\hat{e}-e)\mathbf{1}\{\hat{k}\leq \hat{\gamma}\}+\mu_1\hat{e}(\mathbf{1}\{\hat{k}\leq \hat{\gamma}\}-\mathbf{1}\{\bar{k}\leq \bar{\gamma}\}).
		\label{}
	\end{align}
	Finally, substituting this back into \eqref{eq:r2}, we find that \eqref{eq:r2} is 
	\begin{align}
		&= \P\left[ \frac{(\hat{\mu}_1-\mu_1)(\hat{e}-e)\mathbf{1}\{\hat{k}\leq \hat{\gamma}\}+\mu_1\hat{e}(\mathbf{1}\{\hat{k}\leq \hat{\gamma}\}-\mathbf{1}\{\bar{k}\leq \bar{\gamma}\})}{\hat{e}}\right],\\
		&=\P\left[ \frac{(\hat{\mu}_1-\mu_1)(\hat{e}-e)}{\hat{e}}\mathbf{1}\{\hat{k}\leq \hat{\gamma}\} \right]+\P\left[ \mu_1(\mathbf{1}\{\hat{k}\leq \hat{\gamma}\}-\mathbf{1}\{\bar{k}\leq \bar{\gamma}\}) \right],\\
		&=\P\left[ \mu_1(\mathbf{1}\{\hat{k}\leq \hat{\gamma}\}-\mathbf{1}\{\bar{k}\leq \bar{\gamma}\}) \right]+o_P(n^{-1/2}),
		\label{eq:split}
	\end{align}
	where the third equality repeats the bounding of \eqref{eq:norm_product}. 
	Now, \eqref{eq:split} tracks the contribution of treatment terms to \eqref{eq:simplified}; the contribution of the control terms is analogous, so we conclude that \eqref{eq:simplified} is, up to an $o_P(n^{-1/2})$ error term, $$\frac{1}{n}\sum_{i=1}^n \left( \frac{Y_i(1)-\mu_1(X_i)}{e(X_i)}+\frac{Y_i(0)-\mu_0(X_i)}{1-e(X_i)}(1-Z_i) \right)\mathbf{1}\{\bar{k}(X_i)\leq \bar{\gamma}\}\\ +\P\left[ \tau(X_i)(\mathbf{1}\{\hat{k}(X_i)\leq \hat{\gamma}\}-\mathbf{1}\{\bar{k}(X_i)\leq \bar{\gamma}\}\right],$$ completing the proof.
\end{proof}

For the denominators $n_{\hat{A}}$ and $n_{\bar{A}}$, we recall \eqref{eq:denom_diff}, which shows that 
\begin{align}
	\frac{n_{\hat{A}}-n_{\bar{A}}}{n}&=\P_n\mathbf{1}\{\hat{k}(X_i)\leq \hat{\gamma}\}-\P_n\mathbf{1}\{\bar{k}(X_i)\leq \bar{\gamma}\},\\
&= \P(\mathbf{1}\{\hat{k}(X_i)\leq \hat{\gamma}\}-\mathbf{1}\{\bar{k}(X_i)\leq \bar{\gamma}\})+o_P(n^{-1/2}).
\label{eq:denom_diff_2}
\end{align}

Finally, combining Lemma~\ref{lem:numerator} and \eqref{eq:denom_diff_2} gives the theorem. 
\begin{proof}[Proof of Theorem~\ref{thm:limit_trimmed_clt}]
	We apply a first-order Taylor expansion to the function $f(x,y)=x/y$ around $(S_{\bar{A}}, n_{\bar{A}})$ to write $\hat{S}_{\hat{A}}/n_{\hat{A}}$ as 
	\begin{align}
		\frac{\hat{S}_{\hat{A}}}{n_{\hat{A}}}&=\frac{S_{\bar{A}}}{n_{\bar{A}}}+\frac{(\hat{S}_{\hat{A}}-S_{\bar{A}})/n}{n_{\bar{A}}/n}-\frac{S_{\bar{A}}/n}{(n_{\bar{A}}/n)^2}\cdot \frac{n_{\hat{A}}-n_{\bar{A}}}{n}+o_P\left(\|[\hat{S}_{\hat{A}}/n\;\; n_{\hat{A}}/n] - [S_{\bar{A}}/n\;\; n_{\bar{A}}/n]\|_2\right),\\
		&=\frac{S_{\bar{A}}}{n_{\bar{A}}}+\frac{(\hat{S}_{\hat{A}}-S_{\bar{A}})/n}{n_{\bar{A}}/n}-\frac{S_{\bar{A}}/n}{(n_{\bar{A}}/n)^2}\cdot \frac{n_{\hat{A}}-n_{\bar{A}}}{n}+o_P(n^{-1/2}),
	\end{align}
	because $[\hat{S}_{\hat{A}}/n\;\; n_{\hat{A}}/n] - [S_{\bar{A}}/n\;\; n_{\bar{A}}/n]=O_p(n^{-1/2})$ by assumption. Rearranging, we find
	\begin{align}
		\hat{\tau}_{\hat{A}}-\tau_{\bar{A}}&= \frac{(\hat{S}_{\hat{A}}-S_{\bar{A}})/n}{n_{\bar{A}}/n}-\frac{S_{\bar{A}}/n}{(n_{\bar{A}}/n)^2}\cdot \frac{n_{\hat{A}}-n_{\bar{A}}}{n}+o_P(n^{-1/2}),\\
		&= \frac{1}{\P(\bar{k}(X_i)\leq \bar{\gamma})+o_P(1)}\cdot \frac{\hat{S}_{\hat{A}}-S_{\bar{A}}}{n}-\frac{\E[\tau(X_i)\mathbf{1}\{\bar{k}(X_i)\leq \bar{\gamma}\}]+o_P(1)}{(\P(\bar{k}(X_i)\leq \bar{\gamma})+o_P(1))^2}\cdot \frac{n_{\hat{A}}-n_{\bar{A}}}{n}+o_P(n^{-1/2}),\\
	&= \frac{1}{\P(\bar{k}(X_i)\leq \bar{\gamma})}\cdot \frac{\hat{S}_{\hat{A}}-S_{\bar{A}}}{n}-\frac{\E[\tau(X_i)\mathbf{1}\{\bar{k}(X_i)\leq \bar{\gamma}\}]}{\P(\bar{k}(X_i)\leq \bar{\gamma})^2}\cdot \frac{n_{\hat{A}}-n_{\bar{A}}}{n}+o_P(n^{-1/2}),\\
	&= \frac{1}{\P(\bar{k}(X_i)\leq \bar{\gamma})}\cdot \frac{\hat{S}_{\hat{A}}-S_{\bar{A}}}{n}-\frac{\E[\tau(X_i)\mid \bar{k}(X_i)\leq \bar{\gamma}]}{\P(\bar{k}(X_i)\leq \bar{\gamma})}\cdot \frac{n_{\hat{A}}-n_{\bar{A}}}{n}+o_P(n^{-1/2}),
	\label{eq:taylor_diff}
	\end{align}
	where the second equality is the law of large numbers, the third follows from bounding the contribution of the $o_P(1)$ terms, and the fourth is the definition of conditional expectation. Then, substituting Lemma~\ref{lem:numerator} and \eqref{eq:denom_diff_2} into \eqref{eq:taylor_diff} and simplifying gives \eqref{eq:thm2_statement}, as desired.
\end{proof}

\section{Additional experimental results}
\label{app:acic}

\subsection{Visualization of variance minimization on NHANES data} 

In this section, we clarify a subtlety regarding Figure~\ref{fig:nhanes_seq}. Specifically, one might wonder why our chosen cut-off $\hat{\gamma}$ does not minimize the standard errors shown in the bottom panel of Figure~\ref{fig:nhanes_seq}. The reason for this is that $\hat{\gamma}$ is chosen to minimize the objective defined in \eqref{eq:gamma_objective}, not to minimize the actual estimated standard error on the data (that is, the sample variance of the AIPW scores). We plot these two quantities as a function of $\hat{\gamma}$ in Figure~\ref{fig:nhanes_obj} for the case of heteroscedastic trimming to demonstrate the difference between the two. This distinction is in fact crucial, because if we chose $\hat{\gamma}$ to minimize the actual standard errors, we would be using the responses $Y_i$ directly, and our inferences would be invalid. In choosing $\hat{\gamma}$ to minimize the objective in \eqref{eq:gamma_objective}, the responses $Y_i$ are used only indirectly in the fitting of $\hat{k}$, and by preventing $\hat{k}$ from overfitting as we do in Assumption~\ref{ass:dml}, we can obtain valid inferences.

\begin{figure}[h]
	\centering
	\includegraphics{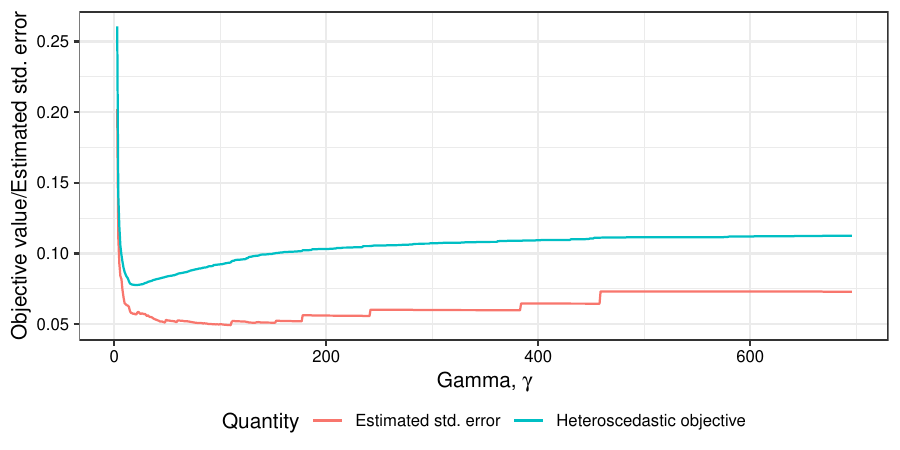}
	\caption{Comparison of the heteroscedastic objective function in \eqref{eq:gamma_objective} that $\hat{\gamma}$ is chosen to minimize (blue) and the actual estimated standard error on the data (red). We see that the minimizer $\hat{\gamma}$ of the objective function has a standard error that is close to the true minimum of the estimated standard errors, but that it trims slightly more units than may actually be necessary.}
	\label{fig:nhanes_obj}
\end{figure}

\subsection{Further experiments on ACIC data}
In this section, we repeat our experiments on the American Causal Inference Challenge data in Section~\ref{subsec:acic} on additional datasets provided by the organizers of the challenge. The full contest consists of 3400 datasets from 17 different data-generating processes, where each data-generating process has a different degree of confounding, propensity score model, and response model \citep{acic_challenge}. Below, we present results on the first 8 of these datasets to highlight that our method has consistently reasonable behavior, even when it does not enable new significant discoveries as in Section~\ref{subsec:acic}.

\begin{figure}[h]
	\centering
	\begin{tabular}{cc}
		\includegraphics[width=0.45\textwidth]{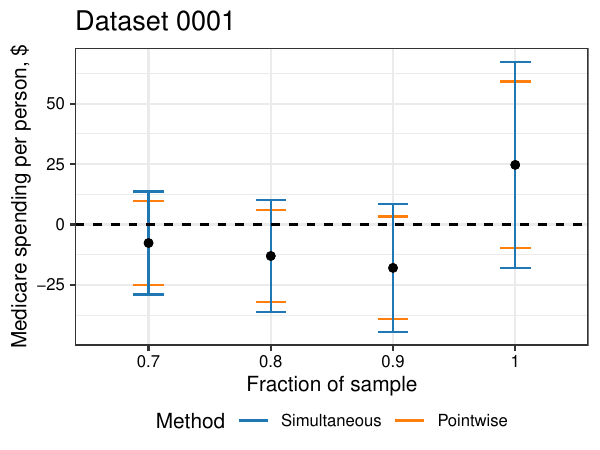}&\includegraphics[width=0.45\textwidth]{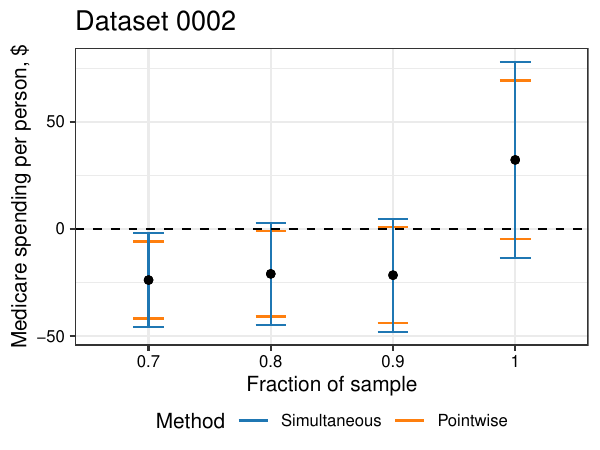}\\
		\includegraphics[width=0.45\textwidth]{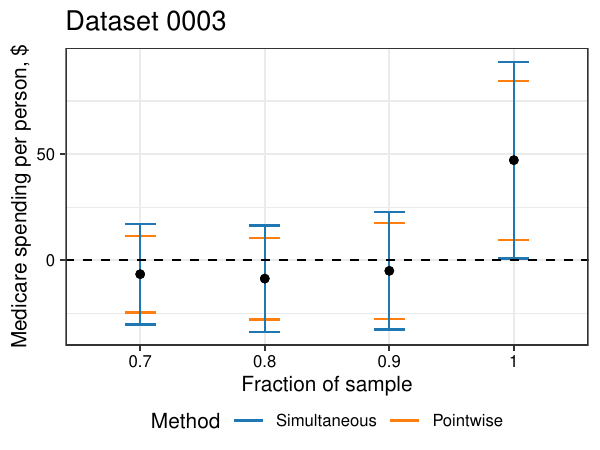}&\includegraphics[width=0.45\textwidth]{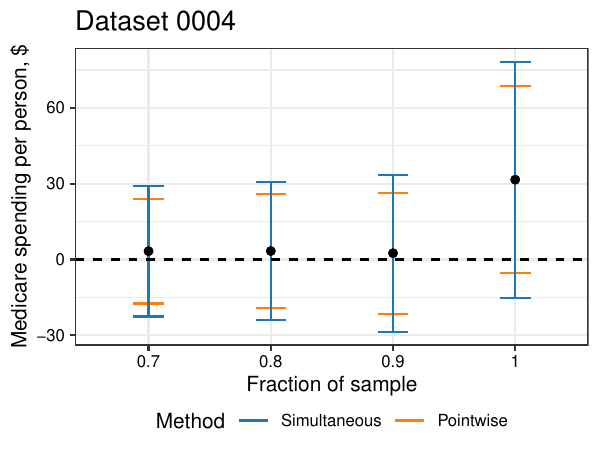}\\
		\includegraphics[width=0.45\textwidth]{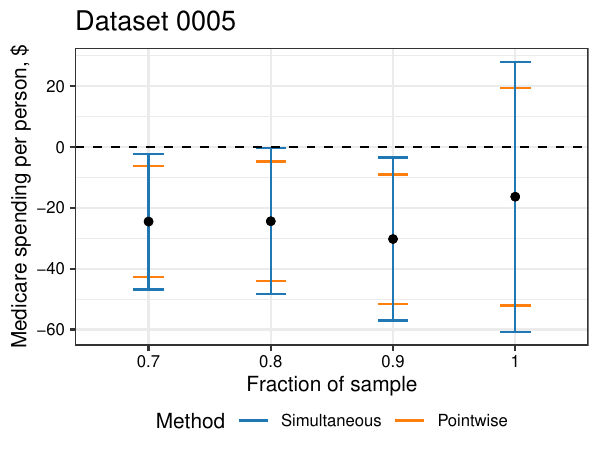}&\includegraphics[width=0.45\textwidth]{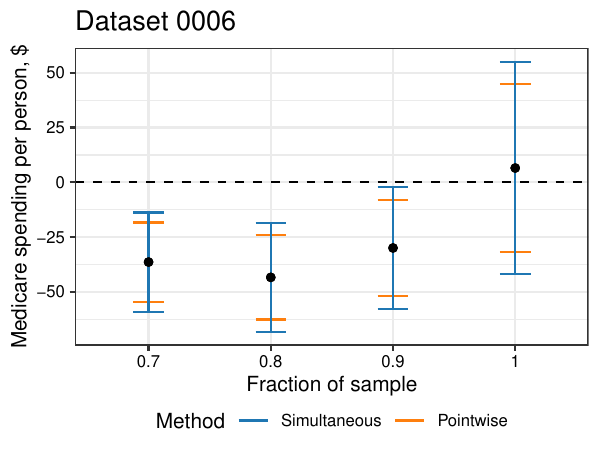}\\
		\includegraphics[width=0.45\textwidth]{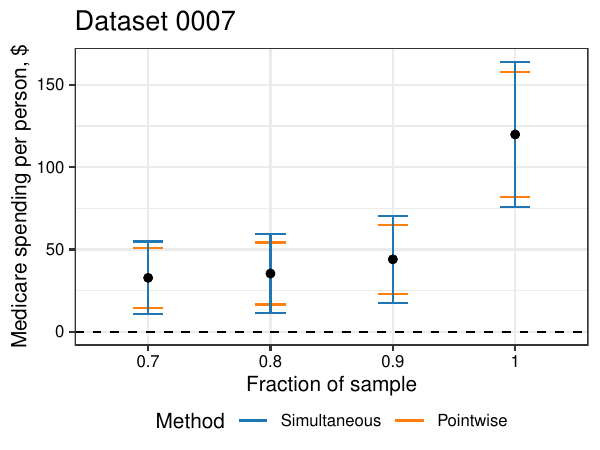}&\includegraphics[width=0.45\textwidth]{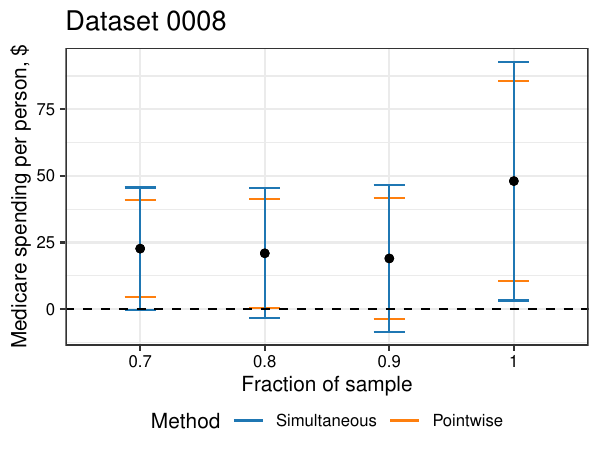}\\
	\end{tabular}
	\caption{Results of repeating the experiments of Section~\ref{subsec:acic} on seven additional datasets (the results of Section~\ref{subsec:acic} correspond to data set six here). We see across all datasets that the simultaneous intervals are only slightly wider than the marginal intervals.}
	\label{fig:acic_replications}
\end{figure}

\end{document}